\documentclass[a4paper]{llncs}

\usepackage[utf8]{inputenc}
\usepackage[USenglish]{babel}

\usepackage{syntax} 
\usepackage{mathpartir} 
\usepackage{amsmath}
\usepackage{amssymb}
\usepackage{mathtools} 
\usepackage{placeins} 
\usepackage{todonotes}
\usepackage{xcolor}
\usepackage{colortbl}
\usepackage{ulem}
\usepackage{alltt}
\usepackage{comment}
\usepackage[nocompress]{cite}
\newcommand{\addr}{\mathrm{addr}} 
\newcommand{\addrof}[1]{\addr(#1)} 

\newcommand{\thread}{\mathrm{thread}} 
\newcommand{\threadof}[1]{\thread(#1)} 

\newcommand{\length}{\mathrm{length}} 
\newcommand{\lengthof}[1]{\length(#1)} 

\newcommand{\delays}{\mathrm{delays}} 
\newcommand{\delaysof}[1]{\delays(#1)} 

\newcommand{\reorders}{\mathrm{reorders}} 
\newcommand{\reordersof}[1]{\reorders(#1)} 

\newcommand{\mm}{\mathrm{\$}} 
\newcommand{\mmof}[1]{\mm(#1)} 

\newcommand{\trace}{\mathrm{Tr}} 
\newcommand{\traceof}[1]{\trace(#1)} 

\newcommand{\psocomp}{C} 
\newcommand{\psocompof}[1]{\psocomp(#1)} 
\newcommand{\sccomp}{C_{\mathrm{SC}}} 
\newcommand{\sccompof}[1]{\sccomp(#1)} 

\newcommand{\aprogram}{\mathcal{P}} 
\newcommand{\athread}{\mathrm{t}} 
\newcommand{\anaction}{a} 
\newcommand{\anotheraction}{b} 
\newcommand{\athirdaction}{c} 
\newcommand{\acc}{\mathrm{acc}} 
\newcommand{\alab}{\textsf{l}}
\newcommand{\areg}{r}
\newcommand{\anaddr}{a}
\newcommand{\anexpr}{e}

\newcommand{\isu}{\mathrm{isu}}
\newcommand{\st}{\mathrm{st}}
\newcommand{\ld}{\mathrm{ld}}
\newcommand{\fence}{\mathrm{fence}}
\newcommand{\scfence}{\mathrm{scfence}}
\newcommand{\loc}{\mathrm{loc}}

\newcommand{\FUN}{\mathrm{FUN}}
\newcommand{\DOM}{\mathrm{DOM}}
\newcommand{\VAR}{\mathrm{VAR}}
\newcommand{\LAB}{\mathrm{LAB}}
\newcommand{\THRD}{\mathrm{THRD}}
\newcommand{\ACT}{\mathrm{ACT}}

\newcommand{\projection}[2]{#1\downarrow{}#2}

\newcommand{\hb}{\rightarrow_{\text{hb}}}
\newcommand{\po}{\rightarrow_{\text{po}}}
\newcommand{\sto}{\rightarrow_{\text{st}}}
\newcommand{\src}{\rightarrow_{\text{src}}}
\newcommand{\cf}{\rightarrow_{\text{cf}}}
\newcommand{\cfst}{\rightarrow_{\text{cf/st}}}
\newcommand{\srcst}{\rightarrow_{\text{src/st}}}
\newcommand{\popluscfst}{\rightarrow_{\text{po$^+$/cf/st}}}

\newcommand{\transition}[1]{\xrightarrow{#1}}

\newcommand{\xaddr}{\textsf{x}}
\newcommand{\yaddr}{\textsf{y}}

\newcommand{\theprogram}[2]{\lit*{program} #1 #2}
\newcommand{\thethread}[4]{\lit*{thread} \ensuremath{#1} \lit*{regs}
\ensuremath{#2} \lit*{init} \ensuremath{#3} \theinstructions{#4}}
\newcommand{\theinstructions}[1]{\lit*{begin} #1\lit*{end}}
\newcommand{\thetransition}[3]{\ensuremath{#1}\lit*{:}\ #3\lit*{;}\ \lit*{goto}\ \ensuremath{#2}\lit*{;}}
\newcommand{\thestore}[2]{\lit*{mem[}\ensuremath{#1}\lit*{]} \ensuremath{\leftarrow} \ensuremath{#2}}
\newcommand{\theload}[2]{\ensuremath{#1} \ensuremath{\leftarrow} \lit*{mem[}\ensuremath{#2}\lit*{]}}
\newcommand{\thecondition}[1]{\lit*{assert}\ \ensuremath{#1}}
\newcommand{\thelocal}[2]{\ensuremath{#1} \ensuremath{\leftarrow} \ensuremath{#2}}
\newcommand{\themem}[1]{\lit*{mem[}\ensuremath{#1}\lit*{]}}
\newcommand{\thescfence}[0]{\lit*{scfence}}
\newcommand{\thefence}[1]{\ensuremath{\lit*{fence}_{#1}}}

\newcommand{\sem}[1]{[\![#1]\!]}
\newcommand{\semattackerstinst}[1]{\sem{#1}_{\attack 1}}
\newcommand{\semattackerlastinst}[1]{\sem{#1}_{\attack 2}}
\newcommand{\semattacker}[1]{\sem{#1}_{\attack 3}}
\newcommand{\semhelperorig}[1]{\sem{#1}_{\textsf{H0}}}
\newcommand{\semhelpertrans}[1]{\sem{#1}_{\textsf{H1}}}
\newcommand{\semhelpercpy}[1]{\sem{#1}_{\textsf{H2}}}

\newcommand{\attack}{\textsf{A}}
\newcommand{\stinst}{\textsf{stinst}}
\newcommand{\lastinst}{\textsf{lastinst}}
\newcommand{\attacker}{\textsf{t}_{\textsf{A}}}
\newcommand{\attackstore}{\textsf{st}_{\textsf{A}}}
\newcommand{\attackerwaitlabel}{\tilde\alab_{\textsf{wait}}}
\newcommand{\attackerwait}{\textsf{wait}}

\newcommand{\loadacc}{\textsf{lda}}
\newcommand{\storeacc}{\textsf{sta}}

\newcommand{\auxdelayed}[1]{(#1, \textsf{d})}
\newcommand{\auxaccesslevel}[1]{(#1, \textsf{hb})}
\newcommand{\auxhb}{\textsf{hb}}
\newcommand{\auxsuc}{\textsf{suc}}
\newcommand{\auxfence}{\areg_\textsf{fence}}
\newcommand{\auxaddress}{\areg_{\attackstore}}
\newcommand{\auxdelayval}{\areg_{\textsf{delayval}}}

\newcommand{\mytrue}{\mathsf{true}}
\newcommand{\maxfun}{\textsf{max}}

\newcommand{\pcconf}{\lit*{pc}}
\newcommand{\valconf}{\lit*{val}}
\newcommand{\bufconf}{\lit*{buf}}

\newcommand{\flag}{\mathit{flag}}

\newcommand{\sschweizer}[1]{\textcolor{blue}{#1}}

\usepackage{footmisc}
\newcommand\blfootnote[1]{%
  \begingroup
  \setlength\footnotemargin{0pt}%
  \renewcommand\thefootnote{}\footnotetext{#1}%
  \endgroup
}

\title{Locality and Singularity\\ for  
Store-Atomic Memory Models}

\author{Egor Derevenetc$^{\text{1,3}}$\qquad Roland Meyer$^{\text{2}}$\qquad Sebastian Schweizer$^{\text{3}}$\vspace{-0.2cm}}
\institute{$^{\text{1}}$Fraunhofer ITWM\qquad $^{\text{2}}$TU Braunschweig\qquad $^{\text{3}}$TU Kaiserslautern}
\begin{document}
\maketitle

\begin{abstract}
Robustness is a correctness notion for concurrent programs running under relaxed consistency models. 
The task is to check that the relaxed behavior coincides (up to traces) with sequential consistency (SC).
Although computationally simple on paper (robustness has been shown to be PSPACE-complete for TSO, PGAS, and Power),
building a practical robustness checker remains a challenge.
The problem is that the various relaxations lead to a dramatic number of computations, only few of which violate robustness.

In the present paper, we set out to reduce the search space for robustness checkers. 
We focus on store-atomic consistency models and establish two completeness results. 
The first result, called locality, states that a non-robust program always contains a violating computation where only one thread delays commands.
The second result, called singularity, is even stronger but restricted to programs without lightweight fences.
It states that there is a violating computation where a single store is delayed. 

As an application of the results, we derive a linear-size source-to-source translation of robustness to SC-reachability. 
It applies to general programs, regardless of the data domain and potentially with an unbounded number of threads and with unbounded buffers. 
We have implemented the translation and verified, for the first time, PGAS algorithms in a fully automated fashion.
For TSO, our analysis outperforms existing tools.

\end{abstract}

\blfootnote{\scriptsize This work was supported by the DFG project \emph{R2M2: Robustness against Relaxed Memory Models}.}

\section{Introduction}
Performance drives the design of computer architectures. 
The computation time of a task depends on the time it takes to access the memory. 
To reduce the access time, the idea is to place the data close to the compute unit. 
This idea is applied on virtually all design layers, from multiprocessors to high-performance computing clusters.
Yet, the realization is different. 
Multiprocessors like Intel's x86~\cite{SewellCACM2010} and Sparc's PSO~\cite{sparc-v9-manual} implement thread-local instruction buffers that allow to execute store commands without waiting for the memory.  
The effect of buffered stores will be visible to other threads only when the multiprocessor decides to batch process the buffer,
thus leading to a reordering of instructions.
Clusters often implent a programming model called partitioned global address space (PGAS), either in terms of APIs like SHMEM~\cite{chapman2010introducing}, ARMCI~\cite{nieplocha1999armci}, GASNET~\cite{bonachea2002gasnet}, GPI~\cite{machado2009fraunhofer}, and GASPI~\cite{GASPI}, or by HPC languages like UPC~\cite{UPC}, Titanium~\cite{hilfinger2005titanium}, and Co-Array Fortran~\cite{numrich1998co}.  
The PGAS model joins the partitioned memories of the cluster nodes into one (virtual) global memory.  
The selling point of PGAS is one-sided communication:  
A thread can modify a part of the global memory that resides in another node, without having to synchronize with that node.
The drawback is the network delay. 
Although already computed, it may take a moment to install a value in the memory of another node. 

Moving the data to the computation is delicate. 
When the data is shared, it has to be split into copies, one copy for each thread holding the datum. 
But then, in the interest of performance, updates to one copy cannot be propagated immediately to the other copies. 
This means the computations have to \emph{relax} the guarantees given by an atomic memory and captured by the notion of sequential consistency (SC)~\cite{Lamport79}. 
The commands no longer take effect on the global memory in program order but may be reordered by the architecture.
An important guarantee of SC, however, remains true in all the above models: \emph{Store atomicity}. 
Once a store command arrives at the global memory, it is visible to all threads.

Programming a shared memory is difficult. 
Having to take into account the reorderings of the architecture makes programming in the presence of relaxed consistency close to impossible. 
SC-preserving compilers have been proposed as an approach to the problem and receive considerable attention~\cite{ShashaSnir88,MidkiffFences2003,scpreserving,AlglaveM11,BMM11,endend12}.
The idea is to let the developer implement for SC, and make it the task of the compiler to insert synchronization primitives that justify the SC-assumption for the targeted architecture. 
Algorithmically, justifying the SC-assumption amounts to checking a problem known as \emph{robustness} (against the architecture of interest): For every relaxed computation there has to be an SC-computation with the same behavior.
The notion of behavior to be preserved typically (and also in this work) is the happens-before traces~\cite{lamport1978time}.  
When developing an SC-preserving compiler, checking robustness is the main task. 
Inferring synchronization primitives from indications of non-robustness is better understood~\cite{Kuperstein2010,AlglaveCAV10,VafeiadisN11,KupersteinVY12,AbdullaACLR12,BDM13,BMC2013,Cycles2014,SAS14,AAP15}.

An SC-preserving compiler needs an over-approximate robustness analysis that should be as precise as possible. 
Under-approximations like bounded model checking~\cite{BurckhardtCaseStudyCAV2006,Burckhardt2007,BMC2013}, simulation~\cite{Cats14}, or monitoring~\cite{burckhardt-musuvathi-CAV08,Sen2011} may miss non-robust computations and insert too few fences.
Over-approximations, if too coarse, lead to over-fencing. 
Although decision procedures for robustness exist~\cite{BMM11,CDMM13,DM14}, building an efficient and yet precise robustness checker remains a challenge.  
The problem is the immense degree of non-determinism brought by the instruction reorderings that is hard to reflect in the analysis.   
This non-determinism forces over-approximations into explicitly modeling architectural details like instruction buffers~\cite{AbdullaACLR12,KupersteinVY11,AlglaveKNT13,VMCAIVechev15} (operational approaches) or right-away operating on the code (axiomatic approaches)~\cite{ShashaSnir88,AlglaveM11,Cycles2014}. 

In this paper, we contribute two semantical results about robustness that limit the degree of non-determinism that has to be taken into account by algorithmic analyses. 
Both results state that robustness violations can be detected --- in a complete way --- with a restricted form of reorderings. 
The first result, called locality, states that only one thread needs to make use of instruction reorderings.
The other threads operate as if they were running under SC:
\emph{A program is not robust if and only if there is a violating computation where exactly one thread delays stores.}
The second result, called singularity, is even stronger:
\emph{A program without lightweight fence instructions is not-robust if and only if there is a violating computation with exactly one delayed store.}
Note that a program without delays is robust.
This means the result is an optimal characterization of non-robustness.
Singularity only holds in the absence of lightweight fences. 
We do not consider this a severe limitation.
Robustness is meant as a subroutine for fence inference inside an SC-preserving compiler.
In that setting, programs naturally come without fences. 


Our third contribution shows that the development of specialized robustness analyses can be avoided.  
Utilizing locality and singularity, we give an instrumentation that reduces robustness to reachability \emph{under SC}. 
By instrumentation, we mean a source-to-source translation of a given program $\aprogram$ into a program $\aprogram'$ so that the former is robust if and only if the latter does not reach under SC a designated state. 
This allows us to employ for the analysis of robustness all techniques and tools that have been developed for SC-reachability. 
As a side-effect, we obtain the decidability of robustness for parameterized programs over a finite data domain.
The restriction to finite data domains is not necessary for the instrumentation itself, but for the back-end SC-reachability analysis.


Concerning the model, we show that locality holds for virtually all store-atomic consistency models (singularity holds in the absence of dependencies).
Inspired by~\cite{ABBM12}, we introduce a programming language for concurrent programs that is meant to act as a programming framework for store-atomic models.
The syntax of our language is an assembly dialect enriched with a variety of fence commands. 
The semantics is defined weak enough to support the relaxations found in the models discussed above.  
What makes our programming language a programming framework is that, given a program, we can add appropriate fences to obtain the behavior under SC, TSO, PSO, and PGAS. 
The motivation for having a programming framework is that we can show locality and singularity once for this model, and it will then hold for all instances of the framework.



For improved readability, this paper contains only a comprehensible high-level explanation of our techniques.
The full technical formalizations and proofs can be found in the appendix.

\section{Related Work}
Robustness checks that the relaxed behavior of a program is the same as the behavior under SC. 
The definition is relative to a notion of behavior, and there are various proposals in the literature~\cite{ShashaSnir88,AbdullaACLR12,AAP15}.  
To make the most of the consistency model in terms of performance, the notion of behavior should be liberal enough to equate quite distinct computations. 
On the other hand, it should be strong enough to be easy to check algorithmically.
Equivalence of the happens-before traces~\cite{ShashaSnir88} appears to be a good compromise between expressiveness and algorithmics, and is the favored notion in the literature on robustness~\cite{burckhardt-musuvathi-CAV08,Sen2011,BMM11,AlglaveM11,BDM13,CDMM13}. 
Abdulla et al. recently proposed an alternative that is incomparable with the happens-before traces~\cite{AAP15}.
The idea is to preserve the total ordering of all stores and drop the relations for loads. 
This equivalence leads to efficient algorithmics for TSO but does not seem to fit well with consistency models beyond TSO. 
State-space equivalence from~\cite{AbdullaACLR12} leads to non-primitive recursive lower bounds for robustness, and will therefore also not be our choice.

\begin{figure}[t]
\newcommand{\ours}{{\cellcolor[gray]{.8}}}
\newcommand{\cons}[1]{(#1)}
\begin{gather*}
\begin{aligned}
\begin{tabular}{l||c|c|c||r|}
			Memory Model\		&\ Ordered \ &\ Locality \ &\ Singularity\ &\ Cost Function \ \\	
    \hline\hline
	TSO	&			\cons{$\surd$}	\phantom{\cite{DM14}}	& \cons{$\surd$} \cite{BMM11}			& \ours{\textsf{X}}	& Delays  \\ \hline
	PGAS $-$ $\texttt{fence}$\		& \cons{$\surd$} \phantom{\cite{DM14}} & \cons{$\surd$} \phantom{\cite{BMM11}} &				\ours{$\surd$} & \ours{DRL (Sec.~\ref{Section:minimal-violations})}\\ \hline
	PGAS $+$ $\texttt{fence}$\		& \cons{$\surd$} \cite{CDMM13} & \phantom{(}\ours{$\surd$}\phantom{) \cite{BMM11}} &				\ours{\textsf{X}}	& \ours{DRL (Sec.~\ref{Section:minimal-violations})}\\ \hline
    Power\		& \phantom{(}$\surd$\phantom{)} \cite{DM14} & \phantom{(}?\phantom{) \cite{BMM11}} &	\ours{\textsf{X}}	& Length	\\ 
    \hline
\end{tabular}
\end{aligned}\\
	\begin{aligned}
		\surd&:\,	\text{holds} &
		\textsf{X}&:\, \text{fails} &
		?&:\, \text{open}&
	\end{aligned}\\[-.3em]
	\begin{aligned}
			(-)&:\, \text{implied from right or below}&
		\text{shaded}&:\, \text{this paper} 	\end{aligned}
\end{gather*}
\vspace{-0.5cm}
\caption{Normal-form results for violating computations under relaxed consistency.}
\label{Figure:Results}
\end{figure}

In our earlier work on TSO~\cite{BDM13}, PGAS~\cite{CDMM13}, and Power~\cite{DM14}, we established normal form results similar to locality and singularity and also made use of the combinatorial proof principle. 
We elaborate on why the reasoning in this paper is substantially different from our earlier efforts. 
Figure~\ref{Figure:Results} summarizes the comparison. 
The results about PGAS~\cite{CDMM13} and Power~\cite{DM14} rely on a normal form (Ordered in Figure~\ref{Figure:Results}) for violating computations where all threads delay commands. 
As the normal form gives weak guarantees, it can be established with a cost function that simply measures the length of violating computations. 
The value of these earlier results is in that they apply to all consistency models which forbid out-of-thin-air values, while giving the precise complexity (but no useful algorithms).
For TSO, we proved locality in~\cite{BMM11}. 
As TSO architectures have one buffer per thread, we were able to apply a cost function that only measures delays, i.e., the number of commands that are processed while a store is being buffered (Figure~\ref{Figure:Results}).
In this paper, we also have to account for overtakes of stores in different buffers.
Another aspect is that, for TSO, we managed to leave the happens-before trace unchanged.
We tried to lift this strategy to PGAS but failed.
Instead, we show how to construct a different computation that still is a violation.
Taking apart the computation was a big step.

The instrumentation we present in Section~\ref{Section:Instrumentation} is related to our work on TSO~\cite{BDM13}. 
To be precise, we adapt the instrumentation of the attacker (that delays commands) to the assumption of singularity and to the more relaxed consistency model. 
The instrumentation of the helper threads (that close the happens-before cycle) is the one from ESOP'13.
Since our earlier work assumes locality, the new instrumentation (for singularity) is more compact (we save a linear number of auxiliary addresses). 
An alternative instrumentation-based robustness analysis is presented in~\cite{AlglaveKNT13}.
The instrumentation precomputes an optimized cycle and then checks robustness wrt. that cycle.
The idea can be understood as trading one complex verification task (robustness) for a number of tasks (robustness wrt. a cycle) that can be solved more efficiently. 
A strong point is that the approach is general enough to apply to non-store-atomic consistency models, including Power.
For TSO, PSO, and PGAS, also that instrumentation will benefit (in terms of size) from locality and singularity.  
Related is also the instrumentation in~\cite{AAP15}.
Atig et al. only have to add two variables to the program, which means the instrumentation is more compact than ours. 
But, as explained above, the approach does not seem to be generalizable beyond TSO. 

Instrumentations that mimic the effect of instruction buffers can also be found in reachability analyses for relaxed consistency models.
Bouajjani et al. apply the idea of bounded context switching to TSO~\cite{ABP11}.
Vechev et al.~\cite{VMCAIVechev15} precompute the utilization of the instruction buffers in TSO and PSO, and show how to mimic them under SC without having to shift content between variables.
Common to both works is that they eagerly simulate the effect of buffers. 
In~\cite{BCDM15}, we show how to introduce store buffering lazily, and only where needed to satisfy a TSO reachability query:
The idea is to guide the store buffering by non-robust computations. 
Hence, also reachability analyses benefit from the improvements for robustness presented here.

In this paper, we focus on store-atomic memory models.
Although store atomicity feels like a natural requirement, important multiprocessors like Power~\cite{SarkarPLDI2011,Cats14,DM14} or ARM~\cite{Cats14} are known to be non-store atomic. 
There, stores to independent variables may arrive in independent threads in a different order.
What remains true is coherence. 
All threads see the stores to each variable in the same order.
A major challenge for future work is to understand whether a locality result can be established for Power, potentially under mild assumptions on the programs.
With the results in this paper, singularity can be shown not to hold, Figure~\ref{Figure:Results}.
\section{Locality and Singularity on an Example}\label{Section:Illustration}
We illustrate the concept of robustness and the main results on singularity and locality with the help of message passing, an idiom common in shared memory concurrent programming.
The code is as follows:
\vspace{0.2cm}
\begin{center}
\begin{minipage}{5.5cm}
\begin{alltt}
\theprogram{MessagePassing}{
\thethread{\athread_w}{}{l_0}{
\thetransition{l_0}{l_1}{\thestore{d_1}{1}}
\thetransition{l_1}{l_2}{\thestore{d_2}{1}}
\thetransition{l_2}{l_3}{\thestore{\flag}{1}}
}}
\end{alltt}
\end{minipage}
\hspace{0.5cm}
\begin{minipage}{5.5cm}
\vspace{0.2cm}
\begin{alltt}
\thethread{\athread_r}{r}{l_x}{\thetransition{l_x}{l_y}{\theload{r}{\flag}}
\thetransition{l_y}{l_{z}}{\thecondition{r = 1}}
\thetransition{l_{z}}{l}{\theload{r}{d_1}}
}
\end{alltt}
\end{minipage}
\end{center}
\vspace{0.2cm}
The task of the idiom is to signal to thread $\athread_r$ (called reader) that the data written by thread $\athread_w$ (called writer) is ready for use.
The idiom is useful when the data to be written cannot be updated atomically. 
The situation is as follows. 
The writer stores data in the addresses $d_1$ and $d_2$. 
To hand over the data to the reader, $\athread_w$ raises $\flag$ (initially, all addresses are $0$).
The reader reads the flag and finds it raised.
From this, it concludes that $d_1$ has a new value and reads it. 

For message passing to work, we need the following guarantee: Before the flag is raised, the stores to $d_1$ and $d_2$ have to be visible in memory.
The guarantee holds for SC and TSO but does not hold for PGAS.
Indeed, without further synchronization the above code will fail on a cluster.  
The stores to $d_1$ and $d_2$ may be delayed due to network latencies (they may be big packages).
While these stores are being processed, the flag is successfully set. 
The reader sees the flag and reads the old (or broken due to an incomplete transfer) value at $d_1$.

The message passing idiom is not robust. 
It works properly under SC but fails when ported to PGAS.
Formally, robustness requires that for every computation under PGAS there is a computation under SC that has the same happens-before trace. 
We will show that this is not the case. 
Consider the computation 
\begin{equation*}
\tau = \isu_a\cdot\isu_b\cdot \isu_c\cdot c\cdot d\cdot e\cdot f\cdot b\cdot a,
\end{equation*}
where $\isu_a$ gives the moment when action $a$ is issued and $a$ itself gives the moment when the command is executed on memory.
The meaning of the actions is given in Figure~\ref{Figure:Trace}, which shows the happens-before trace $\traceof{\tau}$ of the computation.
The happens-before trace reflects the crucial dependencies among the actions: The program order, the store order, and the source and conflict relations between loads and stores. 
Computations with the same happens-before trace only differ in how they shuffle independent action.
For computation $\tau$, we see that the flag raised by action $c$ is read by action $d$ (indicated by the source relation), but the load of $d_1$ in action $f$ is overwritten by the later action $a$ (conflict relation).   
\begin{figure}[t]
\begin{center}
\begin{tikzpicture}[nodes={rectangle,draw=none,fill=none}]
  \matrix[row sep=0.3cm,column sep=2cm,nodes={rectangle,draw=none,fill=none}] {
     \node (stored1) {$a\colon$ \thestore{d_1}{1}}; & \node (loadf) {$d\colon$ \theload{r}{\flag}}; \\
 \node (stored2) {$b\colon$ \thestore{d_2}{1}}; & \node (assert) {$e\colon$ \thecondition{r = 1}}; \\
 \node (storef) {$c\colon$ \thestore{\flag}{1}}; & \node (loadd1) {$f\colon$ \theload{r}{d_1}}; \\
  };
\draw[->] (stored1) edge node[midway,left] {$\mathit{po}$} (stored2);
\draw[->] (stored2) edge node[midway,left] {$\mathit{po}$} (storef);
\draw[->] (loadf) edge node[midway,right] {$\mathit{po}$} (assert);
\draw[->] (assert) edge node[midway,right] {$\mathit{po}$} (loadd1);
\path[->] (loadd1.north west) edge node[near end,above] {$\mathit{cf}$} (stored1.south east);
\path[->] (storef.north east) edge node[near end,above] {$\mathit{src}$} (loadf.south west);
\end{tikzpicture}
\end{center}
\caption{Happens-before trace $\traceof{\tau}$ of computation $\tau$.
}
\label{Figure:Trace}
\end{figure}
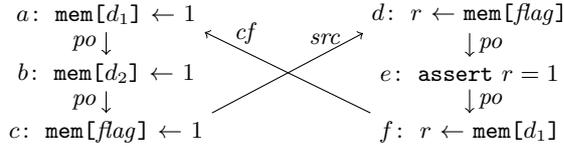

The load at $f$ obtains an old value although the flag was raised.
Under SC where commands take immediate effect this is impossible. 
Formally, the fundamental lemma of Shasha and Snir states that for a relaxed computation $\tau$ there is an SC-computation $\sigma$ with $\traceof{\tau}=\traceof{\sigma}$ if and only if $\traceof{\tau}$ is acyclic~\cite{ShashaSnir88}.
As the above $\traceof{\tau}$ is cyclic, computation $\tau$ has no SC-equivalent. 
Hence, message passing is not robust.

Our singularity result shows that the above computation is unnecessarily complicated.
The theorem states that we need to delay only one action in order to find a robustness violation (a computation with a cyclic trace).
In the example, we can avoid the delay of action $b$.
Computation
\begin{align*}
\tau' = \isu_a\cdot\isu_b\cdot b\cdot \isu_c\cdot c\cdot d\cdot e\cdot f\cdot a
\end{align*}
has the same (cyclic) happens-before trace as $\tau$ but only delays $a$.

To render formally the intuitive feeling that computation $\tau'$ is less relaxed than $\tau$, we associate with each computation a cost. 
The cost of computation $\tau$ is $\mmof{\tau}=(6, 3, 9)$.
There are $6$ \emph{delays}.  
Indeed, $a$ is delayed past $\isu_b\cdot \isu_c\cdot c\cdot b$ and $b$ is delayed past $\isu_c\cdot c$.
There are $3$ \emph{reordering}, action $a$ is issued before $b$ and $c$ but hits the memory later. Similarly, $b$ gets reordered past $c$. 
Altogether there are $9$ actions.
For computation $\tau'$, we have $\mmof{\tau'}=(4, 2, 9)$.
We compare the costs lexicographically and find $\tau'$ less relaxed. 

\section{Concurrent Programs}
\paragraph*{Syntax}
The syntax of our programming language is defined below.  
A concurrent program is identified by a name and consists of a finite set of named threads.  
The threads share a global memory. 
Moreover, each thread defines a finite set of local registers. 
The  code is given as a finite set of labelled instructions.
Each instruction includes a command and the label of the instruction to be executed next. 
To model non-deterministic choices, several instructions can have the same label.
The instruction set includes loads from memory, stores to memory, local assignments, asserts, and two kinds of fences.
SC-fences forbid relaxations altogether.
The second fence command is parameterized by a set of addresses.
To be more precise, the program comes with a domain $\DOM$ the elements of which model the data values as well as the addresses in the global memory.
We assume the domain contains a distinguished value $0 \in \DOM$ that will be used for initialization purposes.
Besides $\DOM$, there is also a function domain $\FUN$ that contains elements from $\DOM^* \rightarrow \DOM$.
All functions that are used in expressions have to stem from $\FUN$.

\vspace{0.1cm}
{\footnotesize
\setlength{\grammarindent}{5em}
\setlength{\grammarparsep}{\parskip}
\begin{center}
\begin{minipage}{6cm}
\begin{grammar}
  <prog> ::= "program" <pid> <thrd>$^*$

  <thrd> ::= "thread" <tid>\\"regs" <reg>$^*$\\"init" <label>\\"begin" <linst>$^*$ "end"

  <linst> ::= <label>":" <inst>; "goto" <label>";"

\end{grammar}
\end{minipage}
\hspace{0.5cm}
\begin{minipage}{5cm}
\begin{grammar}
  <inst> ::= <reg> $\leftarrow$ "mem["<expr>"]"
  \alt "mem["<expr>"]" $\leftarrow$ <expr>
  \alt <reg> $\leftarrow$ <expr>
  \alt "assert" <expr>
  \alt "scfence"
  \alt "fence" <expr>$^*$
  
  <expr> ::= <fun>"(" <reg>$^*$ ")"
\end{grammar}
\end{minipage}
\end{center}}

\vspace{0.4cm}

\paragraph*{Semantics}
Under SC, a store command takes immediate effect on the global memory. 
We consider consistency models that relax the program order but preserve store atomicity. 
This means the effect of a store command may not become visible immediately (to the other threads), but later commands may overtake the store and hit the memory earlier than the store.
What is guaranteed, however, is that once the store is visible it is visible to all threads.

The semantics specifies $\psocompof{\aprogram}$, i.e., the set of computations for program $\aprogram$.
We define the relaxed semantics of our assembly language in an operational style, in terms of a hardware architecture that processes instructions.
In this model, out-of-program-order computations result from the use of instruction buffers inside the architecture. 
To provide an umbrella for different store-atomic models, 
the architecture has two types of buffers.
Buffers of the first type are per thread and per address FIFO buffers that only hold stores of one thread to one address.
Buffers of the second type are per thread FIFO buffers that hold stores of one thread to potentially different addresses. 
The per-address buffers emulate PGAS. 
An all-addresses buffer mimics TSO. 

When a thread issues a store, the instruction is put into the corresponding per-address buffer.
From the per-address buffers, the store non-deterministically advances to the all-addresses buffer of that thread.
From the all-addresses buffer, the store eventually arrives at the global memory.
Due to the per-address buffers, stores of the same thread to different addresses 
may enter the memory in an order different from the order in which they were issued.  

When doing a load, a thread checks whether there is a store to the address of interest kept in one of the store buffers. 
If not, the thread loads the current value from the main memory.
If so, the thread loads the value of the most recent store in the buffers (where stores in the per-address buffers are more recent than stores in the all-addresses buffer). 
The definition ensures that the thread sees its own stores in issue order. 
Even more, for a sequential program the architecture appears to implement SC.

To synchronize different threads, the language offers two fence instructions. 
The $\texttt{scfence}$ instruction can only be executed if all buffers of the executing thread are empty. 
Additionally, we have the $\texttt{fence}$ instruction that carries a list of addresses. 
It can be executed if the buffers for the given addresses in the executing thread are empty.

We show that the framework encompasses the consistency models we aim at.\\[0.1cm]
{\bf Sequential consistency}\quad 
Lamport's SC~\cite{Lamport79} is the intuitive consistency model that reflects an atomic shared memory.  
Formally, the SC-computations are the valid interleavings of the computations of all threads. 
We can mimic SC in our framework by letting stores directly go to memory. 
To enforce this, one can insert an $\texttt{scfence}$ after each store. 
We avoid this modification and just write $\sccompof{\aprogram}$ to mean the set of SC-computations of program $\aprogram$.
\\[0.2cm]
{\bf Total Store Ordering}\quad 
TSO implements the store-to-load relaxation. Each thread has a single buffer for stores, and loads may overtake buffered stores when early reads fail. 
Notably, the thread-local total order of stores is preserved. 
Intel's x86 architecture implements TSO \cite{SewellCACM2010}.

The all-addresses buffers in our framework are meant to mimic TSO. To preserve the total order of stores, the per-address buffers must not be used. 
We can insert $\texttt{fence}$ instructions after each store command to enforce TSO.\\[0.2cm]
{\bf Partial Store Ordering}\quad 
For PSO~\cite{sparc-v9-manual}, there are two kinds of relaxation, the store-to-load relaxations of TSO and the store-to-store relaxation. The latter means that stores of the same thread to different addresses can be reflected in the global memory in an order that is different from the program order.

The per-address-buffers enable the additional store-to-store relaxation. 
As long as a store resides in its per-address buffer, stores to other addresses can be executed faster and the buffered store gets delayed. 
A program that is executed unmodified (without further fences) in our framework will run with PSO semantics. 
The benefit of PSO is that it approximates well the behavior of PGAS. \\[0.2cm]
{\bf Partitioned Global Address Space}\quad 
In a PGAS clusters, the cluster nodes create FIFO buffers to transfer data to and request data from neighboring nodes.
The transfer itself is handled by the network infrastructure.
PGAS APIs allow the user to specify the buffer that should be used for a transfer. 
A comparison and precise model of PGAS APIs can be found in~\cite{CDMM13}.
The model presented here is designed as an approximation of PGAS that is less complex and hence  easier to handle wrt. the theory we develop.  
Rather than assigning stores to buffers, we give each store a separate buffer, like in PSO. 
The user defined sharing of buffers is mimicked by the parameterized fences.
The approximation is a bit weak as it forces different buffers into a total ordering.

\section{Robustness}
We first define the happens-before relation and, based on this, the notion of robustness.
As a first step towards locality and singularity, we then work out properties of computations that violate robustness.

Given $\tau\in\psocompof{\aprogram}$, the happens-before relation highlights the crucial control and data-flow dependencies in the computation. 
Crucial means that an alternative computation with the same relations and valid delays is guaranteed to be executable in the program.
In particular, all asserts will receive the same values.
Technically, the happens-before relation $\hb\ =\ \po \cup \sto \cup \src \cup \cf$ is the union of different dependencies between two actions. 
The program order $\po$ denotes the control flow between actions of the same thread. 
The remaining orders can be summarized as follows.
Whenever two actions $x$ and $y$ operate on the same address in the shared memory, follow each other in $\tau$, and at least one of them is a store, then we have $x\hb y$.
It is common to distinguish the relations according to the commands. 
The store order $\sto$ denotes the order in which two stores to the same address reach the memory. 
The source relation $\src$ goes from a store action to a load action of the same address and denotes the fact that the load reads the value written by that store. 
The conflict relation $\cf$ is a derived relation and says that a load has to happen before a store that potentially overwrites the value to be read.

The \emph{happens-before trace} $\traceof{\tau}$ associated with computation $\tau\in\psocompof{\aprogram}$ is the directed graph where the nodes are the actions from $\tau$.
To be precise, $\isu_x$ and store or fence $x$ are represented by the same node $x$.
The edges are given by the four happens-before relations. 
Note that each computation induces a trace but several computations can have the same trace.

The \emph{robustness problem} is to check, given a concurrent program, whether the traces obtained from the computations in the model from the previous section are included in the traces of the SC computations:
\begin{quote}
Given program $\aprogram$, does $\traceof{\psocompof{\aprogram}}\subseteq \traceof{\sccompof{\aprogram}}$ hold?
\end{quote}
Algorithmically, the task is to look for \emph{violations} of robustness, computations $\tau\in\psocompof{\aprogram}$ with $\traceof{\tau} \notin \traceof{\sccompof{\aprogram}}$. 
Shasha and Snir observed that violating computations are precisely those with a cyclic happens-before relation.
\begin{lemma}[\hspace{-1sp}\cite{ShashaSnir88}]
Consider $\tau\in\psocompof{\aprogram}$. Then 
$\traceof{\tau} \in \traceof{\sccompof{\aprogram}}$ if and only if $\traceof{\tau}$ is acyclic.
\end{lemma}
To see this, consider an acyclic trace.
It will have a linearization that forms an SC computation.
In turn, every SC computation will have an acyclic trace.
Under SC, all commands execute atomically and there is no chance for delays. 
The following development can be understood as strengthening the insight of Shasha and Snir. 
\subsection{Minimal Violations}\label{Section:minimal-violations}
To deepen our understanding of violating computations, we will concentrate on violations that are minimal in a carefully chosen order.
The main finding is the following. 
In a minimal violation, \emph{every delay is due to a cycle}. 
To derive this fact, we employ an interesting proof strategy that establishes properties by contradiction and that will occur in variants throughout the paper. 
Starting from a minimal violation, we assume the property of interest would not hold and from this deduce the existence of a smaller violation.

Technically, we define minimal violations with the help of a cost function $\mmof{-}$ mapping computations to a well-founded domain. 
Intuitively, the cost of a computation reflects its degree of relaxation. 
A minimal violation is then a violation that is as little relaxed as possible (while being a violation). 
Phrased differently, the computation is as close to SC as possible. 
Hence, the cost can also be understood as a penalty for deviating from SC.

The thing to note about the definition of cost is that we define it on the per-thread computations. 
This means that two computations $\tau_1,\tau_2\in\psocompof{\aprogram}$ with the same per-thread computations, $\projection{\tau_1}{\athread}=\projection{\tau_2}{\athread}$ for all threads $\athread$, will have the same cost, $\mmof{\tau_1}=\mmof{\tau_2}$.
Like the proof strategy, this equality will be applied over and over again.
It allows us to choose between different interleavings while preserving minimality.

Technically, the \emph{cost} of a computation is a triple of natural numbers:
\begin{align*}
\mmof{\tau} := (\delaysof{\tau}, \reordersof{\tau}, \lengthof{\tau})\in \mathbb{N}^3.
\end{align*}
We refer to the three auxiliary functions as penalty functions and define them below. 
Cost triples in $\mathbb{N}^3$ are compared lexicographically, so $(4, 0, 6)<(4,1,5)$. 
When we refer to a \emph{minimal violation}, we mean a violation $\tau$ where the cost $\mmof{\tau}$ is minimal in the set of violating computations. 

The intuitive meaning of the penalty functions is as follows. 
The function $\delaysof{-}$ increases when actions happen between an issue and the corresponding store or fence action in memory.
Such intermediary actions indicate a delay of the store. 
Since delays are impossible under SC, there is a penalty for them.
To be precise, we only consider intermediary actions from the thread that executed the store or fence. 
The function $\reordersof{-}$ gives a penalty to reordering delayed stores.
It increases when stores reach the memory in an order different from the order in which they were issued.
Keeping this function value small forces stores of the same thread to different variables to respect the program order.
Finally, function $\lengthof{-}$ gives the computation's length as we would like to focus on violations that do not contain unnecessary actions. We turn to the formalization. 

The \emph{length} of a computation $\tau$, denoted by $\lengthof{\tau}$, is the number of actions in $\tau$. 
The special case $\varepsilon$ is defined to have length zero.

To define the \emph{number of delays}, consider the computation $\tau = \tau_1 \cdot \isu_{\anaction} \cdot \tau_2 \cdot \anaction \cdot \tau_3$ where $\anaction$ is a delayed store or fence action. 
Let $\threadof{\anaction} := \athread$. 
We say that $\anaction$ \emph{overtakes} every action in $\tau_2$ projected to $\athread$.
Hence, for this one store or fence the number of delays is $\delaysof{\anaction} := \lengthof{\projection{\tau_2}{\athread}}$.
The number of delays in $\tau$ is the sum of the delays of all stores and fences:
\begin{align*}
\delaysof{\tau} := \sum_{\text{all stores and fences }\anaction\text{ in }\tau}{\delaysof{\anaction}} \ .
\end{align*}

A store or fence $\anaction$ is said to be \emph{reordered} if it overtakes another issue $\isu_\anotheraction$ together with the corresponding store or fence $\anotheraction$.
Here, $\anaction$ and $\anotheraction$ are supposed to have the same thread $\threadof{\anaction}=\athread=\threadof{\anotheraction}$.
So we have 
\begin{align*}
\tau = \tau_1\cdot \isu_\anaction\cdot \tau_2\cdot \anaction\cdot \tau_3\quad \text{with}\quad \tau_2=\tau_{2a}\cdot\isu_\anotheraction\cdot \tau_{2b}\cdot \anotheraction\cdot \tau_{2c}\ .
\end{align*}
The \emph{number of reorders} for the store $\anaction$,
denoted by $\reordersof{\anaction}$, is the number of such stores or fences $\anotheraction$ in $\projection{\tau_2}{\athread}$.  
The number of reorders in a computation $\tau$, $\reordersof{\tau}$, is again the sum of the reorders of all stores and fences in $\tau$.

\subsection{Cycles}
Our goal is to establish the following result about minimal violations.
Whenever we have an action $\anaction$ that has been overtaken by another action $\anotheraction$, then we already find a cycle involving the two.
To be more precise, either the cycle is via the intermediary actions between $\anaction$ and $\anotheraction$ (Proposition~\ref{hb-cycle-existence}(i)), or the cycle is via a dependency from $\isu_{\anotheraction}$ to $\anaction$ (Proposition~\ref{hb-cycle-existence}(ii)).
For the precise statement, we need a stronger variant of the happens-before relation. 

We will often argue that some of the intermediary actions are sufficient to establish the existence of a happens-before path between two actions. 
To make this dependence on the intermediary actions explicit, we recall the happens-before-through relation from~\cite{BMM11}. 
It can be understood as embedding the happens-before relation, or more generally the trace, into the underlying computation, which is a linear structure. 
Consider the computation $\tau=\tau_1\cdot \anaction\cdot \tau_2\cdot\anotheraction\cdot \tau_3\in\psocompof{\aprogram}$. 
We say that \emph{action $\anaction$ happens before action $\anotheraction$ through $\tau_2$}, if there is a subsequence $\anaction_1\ldots \anaction_n$ of $\tau_2$ so that one of 
\begin{align*}
\anaction_i\po^+\anaction_{i+1},\quad
\anaction_i\src\anaction_{i+1},\quad
\anaction_i\sto\anaction_{i+1},\quad\text{or}\quad
\anaction_i\cf\anaction_{i+1}
\end{align*} 
holds for all $0\leq i\leq n$ with $\anaction_0:=\anaction$ and $\anaction_{n+1}:=\anotheraction$.  
We also refer to $\anaction_0\cdot\ldots \cdot \anaction_{n+1}$  as a \emph{happens-before-through chain}. 
\vspace{0.2cm}
\begin{proposition}[Cycles]
\label{hb-cycle-existence}
Let $\tau = \tau_1 \cdot \anaction \cdot \tau_2 \cdot \anotheraction \cdot \tau_3\in \psocompof{\aprogram}$ be a minimal violation where $\anaction$ has been overtaken by $\anotheraction$.
One of the following holds:
\begin{enumerate}
\item[(i)] $\anotheraction \po^+ \anaction$ and $\anaction \hb^+ \anotheraction$ through $\tau_2$
\item[(ii)] $\anaction \po^+ \anotheraction$ and $\tau_1 = \tau_{1a} \cdot \isu_{\anotheraction} \cdot \tau_{1b}$ and $\isu_{\anotheraction} \hb^+ \anaction$ through $\tau_{1b}$.
\end{enumerate}
\textbf{Remark:} Both cases lead to a happens-before cycle with actions in $\tau_1 \cdot \anaction \cdot \tau_2 \cdot \anotheraction$.
\end{proposition}

The proof of Proposition~\ref{hb-cycle-existence} is non-trivial and relies on a strong dichotomy result, Lemma~\ref{dichotomy}.
The lemma is a disjunction $(i)$ or $(ii)$ and should be read as two implications.
If we read it as $\neg (i)$ implies $(ii)$, it states that given a minimal violation two actions can be swapped as long as they are not separated by a happens-before-through chain.
This will allow us to modify a given computation. 
If we read the dichotomy as $\neg(ii)$ implies $(i)$, it states that given a minimal violation whenever we see two actions of the same thread we are already sure to have a happens-before-through chain between them.
\vspace{0.2cm}
\begin{lemma}[Dichotomy~\cite{BMM11}]\label{dichotomy}
In a minimal violation $\tau = \tau_1 \cdot \anaction \cdot \tau_2 \cdot \anotheraction \cdot \tau_3 \in \psocompof{\aprogram}$
\vspace{-0.5cm}
\begin{enumerate}
\item[(i)] $\anaction \hb^+ \anotheraction$ through $\tau_2$ or
\item[(ii)] there is $\tau' = \tau_1 \cdot \tau_{21} \cdot \anotheraction \cdot \anaction \cdot \tau_{22} \cdot \tau_3 \in \psocompof{\aprogram}$ with $\traceof{\tau} = \traceof{\tau'}$, $\projection{\tau}{\athread}$ = $\projection{\tau'}{\athread}$ for every thread $\athread$, and $\tau_{22}$ a subsequence of $\tau_2$.
\end{enumerate}
\end{lemma}
We note that a similar result has been shown to hold for TSO~\cite{BMM11}. 
The currrent setting requires a more subtle cost function on computations.

\vspace{0.2cm}
\begin{corollary}
\label{local-hb-dependence}
In a minimal violation $\tau = \tau_1 \cdot \anaction \cdot \tau_2 \cdot \anotheraction \cdot \tau_3\in \psocompof{\aprogram}$ with $\threadof{\anaction} = \threadof{\anotheraction}$ we have $\anaction \hb^+ \anotheraction$ through $\tau_2$.
\end{corollary}
The corollary is particularly interesting in the setting where action $\anotheraction$ has overtaken action $\anaction$. 
In this case, it states that the overtake was actually required to execute the actions in the following sense. 
The intermediary actions form a happens-before-through chain that prevents $\anotheraction$ from being executed before $\anaction$. 
The existence of this chain renders formally the intuition that minimal violations do not contain unnecessary delays.

To see the corollary, assume there was no happens-before-through chain. 
Lemma~\ref{dichotomy} would allow us to swap the actions $\anaction$ and $\anotheraction$ while preserving the per-thread computations --- in contradiction to the fact that the actions stem from 
the same thread.

\begin{proof}[of Proposition~\ref{hb-cycle-existence}]
Since $\anotheraction$ overtakes $\anaction$, they are from the same thread and thus program-order dependent.  
Furthermore, $\anotheraction$ is a store or a fence. 
If $\anotheraction \po^+ \anaction$, it is immediate to complete the cycle stated in (i): Corollary~\ref{local-hb-dependence} yields $\anaction \hb^+ \anotheraction$ through $\tau_2$.
If $\anaction \po^+ \anotheraction$, we find the issue action $\isu_\anotheraction$ in $\tau_1$, say $\tau_1 = \tau_{1a} \cdot \isu_\anotheraction \cdot \tau_{1b}$. From Corollary~\ref{local-hb-dependence}, $\isu_\anotheraction \hb^+ \anaction$ through $\tau_{1b}$, as required in~(ii). \qed
\end{proof}

\section{Locality}
\begin{theorem}[Locality]
	\label{locality}
	A concurrent program is not robust if and only if there is a violating computation where exactly one thread delays actions.
\end{theorem}
The difficult task is to show completeness.
We can show that we only need to consider minimal violations of the form $\tau = \tau_1 \cdot x \cdot \tau_2 \cdot y \cdot \tau_3 \cdot \st_y \cdot \tau_4 \cdot \st_x \cdot \tau_5$ where $x$ is overtaken by $\st_x$ and $y$ is overtaken by $\st_y$ with $x$ and $y$ being from different threads.
There are two happens-before cycles: one between $x$ and $\st_x$, the other between $y$ and $\st_y$.
The goal is to move $\st_y$ back over $y$ to save one delay while preserving the computation's trace, or at least get another valid computation with a slightly different trace that also contains a happens-before cycle.
This will be the cycle between $x$ and $\st_x$.
Therefore, we get a violation with less delays than the original one.
This contradicts the initial assumption that the chosen computation was a minimal violation.

With locality at hand, we can constrain the shape of minimal violations. 
A similar normal form was presented for TSO in~\cite{BDM13}, and we can adapt it to the current setting with minimal changes.
\vspace{0.3cm}
\begin{proposition}[Witnesses~\cite{BDM13}]\label{witnesses}
Program $\aprogram$ is robust if and only if there is no minimal violation
$\tau = \tau_1 \cdot \isu_\st \cdot \tau_2 \cdot a \cdot \tau_3 \cdot \st \cdot \tau_4\in \psocompof{\aprogram}$, called \emph{witness computation},  that satisfies the following requirements.
\begin{enumerate}\setlength\itemindent{12pt}
	\item[(W1)] Only $\athread_A:=\threadof{\st} = \threadof{a}$ delays actions.
	\item[(W2)] $\projection{\tau_3}{\athread_A} = \varepsilon$.
	\item[(W3)] $\tau_1$ and $\tau_2 \cdot \anaction \cdot \tau_3$ do not contain delayed actions.
	\item[(W4)] For every $\anotheraction$ in $\tau_3 \cdot \st$ we have $\anaction \hb^+ \anotheraction$ through the intermediary actions.
	\item[(W5)] $\tau_4$ only contains delayed stores and fences of $\athread_A$.
\end{enumerate}
\end{proposition}
\vspace{0.2cm}
Following~\cite{BDM13}, thread $\athread_A$ is referred to as the \emph{attacker}, the remaining threads are called \emph{helpers.}
The computation in our running example is a witness: 
\begin{align*}
\tau'=\isu_a\cdot\isu_b\cdot b\cdot \isu_c\cdot c\cdot d\cdot e\cdot f\cdot a\ .
\end{align*}
Here, $\isu_{\st}$ is $\isu_a$ and the overtaken action is $c$.
Part $\tau_3$ is $d\cdot e\cdot f$.

\section{Singularity}
Our second main result shows that it is sufficient to delay only a single store action. 
The theorem only holds in the absence of lightweight fence commands, $\texttt{scfence}$ is still allowed.
We stress that the theorem is optimal in the sense that it characterizes the least relaxation required to violate sequential consistency. 
A program without delayed stores is always robust.
\begin{theorem}[Singularity]\label{singularity}
Consider a program without $\texttt{fence}$. 
It is not robust if and only if there is a violation with exactly one delayed store.
\end{theorem}
The proof is short and derives a contradiction to a strong combinatorial property.
The reasoning is as follows. 
If the program has a violation with one delayed store, it is not robust. 
If the program is not robust, by Proposition~\ref{witnesses} there is a minimal violation that is a witness computation as defined in the previous section.
In this computation, $\tau_4$ only consists of delayed stores and fences of the attacker thread.
We assume the program does not contain $\texttt{fence}$ commands.
This means $\tau_4$ only consists of delayed stores of the attacker.
Now, if more than one store was delayed, we would have a contradiction to the following Proposition~\ref{no-move-of-two-stores} that came as a surprise to us.
A minimal violation will never place two delayed stores next to each other. 
This already concludes the argumentation.
\begin{proposition}[Two Stores]
	\label{no-move-of-two-stores}
	In the absence of $\texttt{fence}$, there is no minimal violation $\tau_1 \cdot \st_1 \cdot \st_2 \cdot \tau_2$ with $\threadof{\st_1} = \threadof{\st_2}$.
\end{proposition}
\section{Instrumentation}\label{Section:Instrumentation}
To check robustness, we have to look for minimal violations. 
Proposition~\ref{witnesses} reduces the search space as we only have to consider minimal violations in witness form.
If the program does not use the $\texttt{fence}$ instruction, then also singularity holds and $\tau_4$ will be empty in all witness computations. 
Our instrumentation is an adaptation of~\cite{BDM13}.\\[0.4cm]
{\bf Attacks}\quad 
An \emph{attack} is a triple $A = (\attacker, \stinst, \lastinst)$, where $\attacker$ is a thread, $\stinst$ is a store instruction of $\attacker$, and $\lastinst$ is a store or load instruction of $\attacker$. 
Note that attacks are syntactic objects and there is a quadratic number of them. 
An attack $A$ is \emph{feasible} if there is a witness computation where $\attacker$ plays the role of the attacker, $\st$ is an instance of $\stinst$, and $\anaction$ is an instance of $\lastinst$. 
With Proposition~\ref{witnesses}, the program is robust if and only if no attack is feasible. 
Given an attack, we now develop an instrumentation that finds a witness for it.

The witness computation of interest has four phases:
\begin{enumerate}
	\item In $\tau_1$, all issued stores are immediately written to the memory. Eventually, the attacker decides to delay a store.
	\item In $\tau_2$, further actions of the attacker happen either without delay or they are delayed until $\tau_4$. The helpers execute arbitrary actions. At some point, the attacker does its last normal action $a$, which is a load or non-delayed store.
	\item In $\tau_3$, the helpers execute only actions that are happens-before-dependent on $a$ while the attacker pauses.
	\item In $\tau_4$, only the attacker's delayed actions get executed.
\end{enumerate}

We use an observation from \cite{BDM13} that limits the information we have to track about the delayed stores.
The argumentation is as follows. 
If there are delayed stores to an address, the attacker will load the last value that was put into the corresponding buffer. 
The helpers will load the last value that was stored in memory. 
Combined with Property~(W3) of witness computations --- $\tau_2 \cdot \anaction\cdot \tau_3$ does not contain delayed actions --- the content of the buffers is not needed. 
All we have to track is two values per address: The current value in memory and the value of the last buffered store, if any. 
When the attacker executes a store that is delayed, it will set the buffered value. 
When the attacker executes a store without delay or the helpers execute a store, it updates the current value in memory. 
Loads from helpers will always read the current value from memory while loads from the attacker will prefer the buffered one, if it is set.


Recall that the values in $\DOM$ act as addresses. 
We extend $\DOM$ as follows: For each $x \in \DOM$ we add auxiliary addresses $\auxdelayed{x}$ and $\auxaccesslevel{x}$. 
The addresses $\auxdelayed{x}$ hold the values of the last buffered stores. 
The addresses $\auxaccesslevel{x}$ are used to track the happens-before-dependencies required by (W4). Here, we rely on the mechanism from~\cite{BDM13}.
Furthermore, we add the auxiliary addresses $\auxhb$ and $\auxsuc$. Flag $\auxhb$ will tell the helpers that Phase~3 has started and they must execute hb-dependent actions. 
Flag $\auxsuc$ indicates a feasible attack.\\[0.2cm]
{\bf Instrumentation of the Attacker for Locality}\quad 
The attacker operates in three modes. 
Initially, instructions are executed under SC (Phase~1). 
Upon execution of $\stinst$, the attacker can decide to delay that store. 
If the store is to address $x$, we save the stored value to $\auxdelayed{x}$ and the address $x$ to an auxiliary register $\auxaddress$. 
The control flow changes to a modified copy of the code for Phase~2.

During Phase~2, when the attacker has to load from address $y$, it will first check whether $\auxdelayed{y}$ is set. 
If so, it reads that value, otherwise it reads the value of $y$. 
A store to address $y$ can either directly go to memory location $y$ or it is delayed and stored to $\auxdelayed{y}$. 
If there already is a buffered value, then the store cannot be done on memory and has to update the buffered value --- due to the FIFO property of buffers. 
Additionally, we have to delay all stores if there was a fence that had to be delayed. 
A fence has to be delayed if there is a buffered value for at least one of its addresses. 
We use an additional register $\auxfence$ as a flag indicating that a fence has been delayed. 
An $\scfence$ action is not permitted in Phase~2 and will lead to a deadlock.

Upon execution of $\lastinst$, the attacker can decide that this is its last action. 
If it is a load, we make sure that there is no buffered value for that address. 
Otherwise, there is no possibility for hb-dependent actions in $\tau_3$. 
We set the $\auxhb$ flag and go to a special wait label. 
If $\lastinst$ is a store, then we have to make sure that it can be stored directly, i.e., there is no delayed fence and no delayed store for that address. 
We execute the store, set the $\auxhb$ flag, and go to the wait label.

The third mode is in the wait label. 
The attacker waits until a helper thread has completed the happens-before chain, in which case it sets the success flag.

Compared to~\cite{BDM13} for TSO, what is new is the handling of non-delayed stores and fences, and the optimized detection of violations by the attacker. \\[0.3cm]
{\bf Optimized Instrumentation of the Attacker for Singularity}\quad 
When the program does not make use of the $\texttt{fence}$ instruction, then singularity holds. 
We can simplify the above instrumentation as follows. 
As only one store is delayed, we can use a single register $\auxdelayval$ instead of the auxiliary addresses $\auxdelayed{x}$ to save the corresponding value. 
Recall that the address is kept in $\auxaddress$.

\section{Discussion}
We have shown that locality and singularity hold for virtually all store-atomic consistency models (singularity holds in the absence of \texttt{fence} instructions). 
Based on this, we have shown how to reduce robustness to reachability under SC.
The study was motivated by our interest in PGAS clusters.
The immediate next question to ask is whether similar results could hold for non-store-atomic models like Power.
The answer is: We do not yet know.
There is, however, a number of proof principles for relaxed consistency models that came out of this investigation.
First, it seems that a dichotomy lemma is helpful:
Whenever there are two actions from the same thread, there already is a happens-before path between them (Corollary~\ref{local-hb-dependence}).
This, in turn, seems to hold whenever we minimize the cost of computations according to the number of delays plus the reorderings. 
We did not have this result in our work on Power~\cite{DM14} but one may be able to lift it.
Another attractive result that one may want to generalize is the delay of neighboring elements, which is forbidden under minimality (Proposition~\ref{no-move-of-two-stores}).
The ultimate goal would be to develop an understanding under which conditions a consistency model would satisfy a property like locality or singularity.

\newpage
\bibliographystyle{plain}
\bibliography{cited}

\begin{thebibliography}{10}

\bibitem{AbdullaACLR12}
P.~A. Abdulla, M.~F. Atig, Y.-F. Chen, C.~Leonardsson, and A.~Rezine.
\newblock Counter-example guided fence insertion under tso.
\newblock In {\em TACAS}, LNCS 7214, pages 204--219. Springer, 2012.

\bibitem{AAP15}
P.A. Abdulla, M.F. Atig, and N.T. Phong.
\newblock The best of both worlds: Trading efficiency and optimality in fence
  insertion for {TSO}.
\newblock In {\em ESOP}, volume 9032 of {\em LNCS}, pages 308--332. Springer,
  2015.

\bibitem{Cycles2014}
J.~Alglave, D.~Kroening, V.~Nimal, and D.~Poetzl.
\newblock Don't sit on the fence --- {A} static analysis approach to automatic
  fence insertion.
\newblock In {\em CAV}, volume 8559 of {\em LNCS}, pages 508--524. Springer,
  2014.

\bibitem{AlglaveKNT13}
J.~Alglave, D.~Kroening, V.~Nimal, and M.~Tautschnig.
\newblock Software verification for weak memory via program transformation.
\newblock In {\em ESOP}, volume 7792 of {\em LNCS}, pages 512--532. Springer,
  2013.

\bibitem{BMC2013}
J.~Alglave, D.~Kroening, and M.~Tautschnig.
\newblock Partial orders for efficient bounded model checking of concurrent
  software.
\newblock In {\em CAV}, volume 8044 of {\em LNCS}, pages 141--157. Springer,
  2013.

\bibitem{AlglaveM11}
J.~Alglave and L.~Maranget.
\newblock Stability in weak memory models.
\newblock In {\em CAV}, volume 6806 of {\em LNCS}, pages 50--66. Springer,
  2011.

\bibitem{AlglaveCAV10}
J.~Alglave, L.~Maranget, S.~Sarkar, and P.~Sewell.
\newblock Fences in weak memory models.
\newblock In {\em CAV}, volume 6174 of {\em LNCS}, pages 258--272. Springer,
  2010.

\bibitem{Cats14}
J.~Alglave, L.~Maranget, and M.~Tautschnig.
\newblock Herding cats: Modelling, simulation, testing, and data mining for
  weak memory.
\newblock {\em ACM TOPLAS}, 36(2):7:1--7:74, 2014.

\bibitem{ABBM12}
M.~F. Atig, A.~Bouajjani, S.~Burckhardt, and M.~Musuvathi.
\newblock What's decidable about weak memory models.
\newblock In {\em ESOP}, LNCS. Springer, 2012.

\bibitem{ABP11}
M.~F. Atig, A.~Bouajjani, and G.~Parlato.
\newblock Getting rid of store-buffers in {TSO} analysis.
\newblock In {\em CAV}, volume 6806 of {\em LNCS}, pages 99--115. Springer,
  2011.

\bibitem{bonachea2002gasnet}
D.~Bonachea.
\newblock {GASNet} specification, v1.1.
\newblock Technical Report UCB/CSD-02-1207, University of California, Berkeley,
  2002.

\bibitem{BCDM15}
A.~Bouajjani, G.~Calin, E.~Derevenetc, and R.~Meyer.
\newblock Lazy {TSO} reachability.
\newblock In {\em FASE}, volume 9033 of {\em LNCS}, pages 267--282. Springer,
  2015.

\bibitem{BDM13}
A.~Bouajjani, E.~Derevenetc, and R.~Meyer.
\newblock Checking and enforcing robustness against {TSO}.
\newblock In {\em ESOP}, volume 7792 of {\em LNCS}, pages 533--553. Springer,
  2013.

\bibitem{BMM11}
A.~Bouajjani, R.~Meyer, and E.~M{\"o}hlmann.
\newblock Deciding robustness against {T}otal {S}tore {O}rdering.
\newblock In {\em ICALP}, volume 6756 of {\em LNCS}, pages 428--440. Springer,
  2011.

\bibitem{BurckhardtCaseStudyCAV2006}
S.~Burckhardt, R.~Alur, and M.~Martin.
\newblock Bounded model checking of concurrent data types on relaxed memory
  models: {A} case study.
\newblock In {\em CAV}, volume 4144 of {\em LNCS}, pages 489--502, 2006.

\bibitem{Burckhardt2007}
S.~Burckhardt, R.~Alur, and M.~Martin.
\newblock Checkfence: checking consistency of concurrent data types on relaxed
  memory models.
\newblock In {\em PLDI}, pages 12--21. ACM, 2007.

\bibitem{burckhardt-musuvathi-CAV08}
S.~Burckhardt and M.~Musuvathi.
\newblock Effective program verification for relaxed memory models.
\newblock In {\em CAV}, volume 5123 of {\em LNCS}, pages 107--120. Springer,
  2008.

\bibitem{Sen2011}
J.~Burnim, C.~Stergiou, and K.~Sen.
\newblock Sound and complete monitoring of sequential consistency for relaxed
  memory models.
\newblock In {\em TACAS}, volume 6605 of {\em LNCS}, pages 11--25. Springer,
  2011.

\bibitem{CDMM13}
G.~Calin, E.~Derevenetc, R.~Majumdar, and R.~Meyer.
\newblock A theory of partitioned global address spaces.
\newblock In {\em FSTTCS}, pages 127--139, 2013.

\bibitem{chapman2010introducing}
B.~Chapman, T.~Curtis, S.~Pophale, S.~Poole, J.~Kuehn, C.~Koelbel, and
  L.~Smith.
\newblock Introducing {OpenSHMEM}: {SHMEM} for the {PGAS} community.
\newblock In {\em PGAS}, page~2. ACM, 2010.

\bibitem{UPC}
UPC Consortium.
\newblock {UPC} language specification v1.2.
\newblock Technical report, 2005.

\bibitem{VMCAIVechev15}
A.~M. Dan, Y.~Meshman, M.~T. Vechev, and E.~Yahav.
\newblock Effective abstractions for verification under relaxed memory models.
\newblock In {\em VMCAI}, volume 8931 of {\em LNCS}, pages 449--466. Springer,
  2015.

\bibitem{DM14}
E.~Derevenetc and R.~Meyer.
\newblock Robustness against {P}ower is {PSPACE}-complete.
\newblock In {\em ICALP}, volume 8573 of {\em LNCS}, pages 158--171. Springer,
  2014.

\bibitem{MidkiffFences2003}
X.~Fang, J.~Lee, and S.~Midkiff.
\newblock Automatic fence insertion for shared memory multiprocessing.
\newblock In {\em SC}, pages 285--294. ACM, 2003.

\bibitem{GASPI}
Global address space programming interface.
\newblock \url{http://www.gaspi.de/}.

\bibitem{hilfinger2005titanium}
P.~N. Hilfinger, D.~O. Bonachea, K.~Datta, D.~Gay, S.~L. Graham, B.~R. Liblit,
  G.~Pike, J.~Zh. Su, and K.~A. Yelick.
\newblock Titanium language reference manual, version 2.19.
\newblock Technical Report UCB/EECS-2005-15, UC Berkeley, 2005.

\bibitem{Kuperstein2010}
M.~Kuperstein, M.~Vechev, and E.~Yahav.
\newblock Automatic inference of memory fences.
\newblock In {\em FMCAD}, pages 111--119. IEEE, 2010.

\bibitem{KupersteinVY11}
M.~Kuperstein, M.~T. Vechev, and E.~Yahav.
\newblock Partial-coherence abstractions for relaxed memory models.
\newblock In {\em PLDI}, pages 187--198. ACM, 2011.

\bibitem{KupersteinVY12}
M.~Kuperstein, M.~T. Vechev, and E.~Yahav.
\newblock Automatic inference of memory fences.
\newblock {\em SIGACT News}, 43(2):108--123, 2012.

\bibitem{lamport1978time}
L.~Lamport.
\newblock Time, clocks, and the ordering of events in a distributed system.
\newblock {\em CACM}, 21(7):558--565, 1978.

\bibitem{Lamport79}
L.~Lamport.
\newblock How to make a multiprocessor computer that correctly executes
  multiprocess programs.
\newblock {\em IEEE Transactions on Computers}, 28(9):690--691, 1979.

\bibitem{machado2009fraunhofer}
R.~Machado and C.~Lojewski.
\newblock The {Fraunhofer} virtual machine: a communication library and runtime
  system based on the {RDMA} model.
\newblock {\em Computer Science-Research and Development}, 23(3-4):125--132,
  2009.

\bibitem{scpreserving}
D.~Marino, A.~Singh, T.~Millstein, M.~Musuvathi, and S.~Narayanasamy.
\newblock A case for an {SC}-preserving compiler.
\newblock In {\em PLDI}, pages 199--210. ACM, 2011.

\bibitem{SAS14}
Y.~Meshman, A.~M. Dan, M.~T. Vechev, and E.~Yahav.
\newblock Synthesis of memory fences via refinement propagation.
\newblock In {\em SAS}, volume 8723 of {\em LNCS}, pages 237--252. Springer,
  2014.

\bibitem{nieplocha1999armci}
J.~Nieplocha and B.~Carpenter.
\newblock {ARMCI}: A portable remote memory copy library for distributed array
  libraries and compiler run-time systems.
\newblock In {\em Parallel and Distributed Processing}, volume 1586 of {\em
  LNCS}, pages 533--546. Springer, 1999.

\bibitem{numrich1998co}
R.~W. Numrich and J.~Reid.
\newblock Co-array {Fortran} for parallel programming.
\newblock In {\em ACM Sigplan Fortran Forum}, volume~17, pages 1--31. ACM,
  1998.

\bibitem{SarkarPLDI2011}
S.~Sarkar, P.~Sewell, J.~Alglave, L.~Maranget, and D.~Williams.
\newblock Understanding {POWER} multiprocessors.
\newblock In {\em PLDI}, pages 175--186. ACM, 2011.

\bibitem{SewellCACM2010}
P.~Sewell, S.~Sarkar, S.~Owens, F.~Z. Nardelli, and M.~O. Myreen.
\newblock {x86-TSO}: a rigorous and usable programmer's model for x86
  multiprocessors.
\newblock {\em CACM}, 53:89--97, 2010.

\bibitem{ShashaSnir88}
D.~Shasha and M.~Snir.
\newblock Efficient and correct execution of parallel programs that share
  memory.
\newblock {\em ACM TOPLAS}, 10(2):282--312, 1988.

\bibitem{endend12}
A.~Singh, S.~Narayanasamy, D.~Marino, T.~Millstein, and M.~Musuvathi.
\newblock End-to-end sequential consistency.
\newblock ISCA, pages 524--535. IEEE, 2012.

\bibitem{VafeiadisN11}
V.~Vafeiadis and F.~Zappa Nardelli.
\newblock Verifying fence elimination optimisations.
\newblock In {\em SAS}, volume 6887 of {\em LNCS}, pages 146--162. Springer,
  2011.

\bibitem{sparc-v9-manual}
D.~Weaver and T.~Germond, editors.
\newblock {\em The SPARC Architecture Manual Version 9}.
\newblock PTR Prentice Hall, 1994.

\end{thebibliography}

\newpage
\appendix
\section{Details on the Semantics}

Consider a program $\aprogram$ consisting of the threads $\THRD = \{t_1, \dots, t_n\}$.
Let $\LAB$ be the set of all labels that are declared in $\aprogram$. 
Assume each thread $t_i$ has the initial label $l_{0,i}$ and defines the registers $\overline{r_i}$. 
We use $\VAR := \DOM \cup \cup_{i \in [1,n]} \overline{r_i}$ to refer to the set of all variables, i.e., the set of all addresses in the global memory and all local registers.

The operational semantics defines a transition relation between states. 
A \emph{state} $s = (\pcconf, \valconf, \bufconf)$ consists of a program counter, a variable valuation, and some buffer content.
The program counter $\pcconf: \THRD \rightarrow \LAB$ tracks for each thread the label of the instruction to be executed next.
The variable valuation $\valconf: \VAR \rightarrow \DOM$ contains the value of each local register and the value of each address in the global memory. 
The buffer content $\bufconf = \bufconf_1 \cup \bufconf_2$ is $\bufconf_1: \THRD \times \DOM \rightarrow \DOM^*$ and $\bufconf_2: \THRD \rightarrow (\DOM^*)^*$. 
The first ones are the per-address buffers in each thread. 
The second are the all-addresses buffers.
They contain the stores (as mappings from addresses to values in $\DOM\times \DOM$) and the fences (as sequences of addresses in $\DOM^*$). 
The \emph{initial state} is defined by $s_0 := (\pcconf_0, \valconf_0, \bufconf_0)$ with $\pcconf_0(t_i) := l_{0,i}$ for all $t_i \in \THRD$, $\valconf_0(x) := 0$ for all $x \in \VAR$, and $\bufconf_1(t, a) := \varepsilon$, $\bufconf_2(t) := \varepsilon$ for all $t \in \THRD, a \in \DOM$.
Here, we use $\varepsilon$ for the empty word.

\begin{figure}
\begin{equation*}
\inferrule{
  instr = r \leftarrow mem[f_a(\overline{r_a})] \\
  a = f_a(\valconf(\overline{r_a})) \\
  \bufconf(t, a) = \beta \cdot v
}{
  (\pcconf, \valconf, \bufconf)
  \transition{(\athread, ld, a, v)}
  (\pcconf', \valconf[r := v], \bufconf)
}\text{~~early read 1}
\end{equation*}\vspace{0.5\baselineskip}
\begin{equation*}
\inferrule{
	instr = r \leftarrow mem[f_a(\overline{r_a})] \\
	a = f_a(\valconf(\overline{r_a})) \\
	\bufconf(t, a) = \varepsilon \\
	\projection{\bufconf(t)}{(a \leftarrow *)} = \beta \cdot (a \leftarrow v)
}{
(\pcconf, \valconf, \bufconf)
\transition{(\athread, ld, a, v)}
(\pcconf', \valconf[r := v], \bufconf)
}\text{~~early read 2}
\end{equation*}\vspace{0.5\baselineskip}
\begin{equation*}
\inferrule{
  instr = r \leftarrow mem[f_a(\overline{r_a})] \\
  a = f_a(\valconf(\overline{r_a})) \\
  \bufconf(t, a) = \varepsilon \\
  \projection{\bufconf(t)}{(a \leftarrow *)} = \varepsilon \\
  v = \valconf(a)
}{
  (\pcconf, \valconf, \bufconf)
  \transition{(\athread, ld, a, v)}
  (\pcconf', \valconf[r := v], \bufconf)
}\text{~~read memory}
\end{equation*}\vspace{0.5\baselineskip}
\begin{equation*}
\inferrule{
  instr = mem[f_a(\overline{r_a})] \leftarrow f_v(\overline{r_v}) \\
  a = f_a(\valconf(\overline{r_a})) \\
  v = f_v(\valconf(\overline{r_v}))
}{
  (\pcconf, \valconf, \bufconf)
  \transition{(\athread, \isu)}
  (\pcconf', \valconf, \bufconf[(t, a) := \bufconf(t, a) \cdot v])
}\text{~~issue store}
\end{equation*}\vspace{0.5\baselineskip}
\begin{equation*}
\inferrule{
	\bufconf(t, a) = v \cdot \beta \\
	\bufconf(t) = \gamma
}{
(\pcconf, \valconf, \bufconf)
\transition{\varepsilon}
(\pcconf, \valconf, \bufconf[(t, a) := \beta][t := \gamma \cdot (a \leftarrow v)])
}\text{~~advance buffer}
\end{equation*}\vspace{0.5\baselineskip}
\begin{equation*}
\inferrule{
  \bufconf(t) = (a \leftarrow v) \cdot \beta
}{
  (\pcconf, \valconf, \bufconf)
  \transition{(\athread, st, a, v)}
  (\pcconf, \valconf[a := v], \bufconf[t := \beta])
}\text{~~store to memory}
\end{equation*}\vspace{0.5\baselineskip}
\begin{equation*}
\inferrule{
  instr = scfence \\
  \bufconf(t) = \varepsilon \\
  \forall a \in DOM . \bufconf(t, a) = \varepsilon
}{
  (\pcconf, \valconf, \bufconf)
  \transition{(\athread, scfence)}
  (\pcconf', \valconf, \bufconf)
}\text{~~scfence}
\end{equation*}\vspace{0.5\baselineskip}
\begin{equation*}
\inferrule{
	instr = fence\ a_1 \dots a_n \\
	\forall 1 \leq i \leq n . \bufconf(t, a_i) = \varepsilon
}{
(\pcconf, \valconf, \bufconf)
\transition{(\athread, \isu)}
(\pcconf', \valconf, \bufconf[t := \bufconf(t) \cdot fence\ a_1 \dots a_n])
}\text{~~issue fence}
\end{equation*}\vspace{0.5\baselineskip}
\begin{equation*}
\inferrule{
  \bufconf(t) = (fence\ a_1 \dots a_n) \cdot \beta
}{
(\pcconf, \valconf, \bufconf)
\transition{(\athread, fence\ a_1 \dots a_n)}
(\pcconf', \valconf, \bufconf[t := \beta])
}\text{~~fence}
\end{equation*}\vspace{0.5\baselineskip}
\begin{equation*}
\inferrule{
  instr = r \leftarrow f(\overline{r})
}{
  (\pcconf, \valconf, \bufconf)
  \transition{(\athread, loc)}
  (\pcconf', \valconf[r := f(\valconf(\overline{r}))], \bufconf)
}\text{~~local assignment}
\end{equation*}\vspace{0.5\baselineskip}
\begin{equation*}
\inferrule{
  instr = assert\ f(\overline{r}) \\
  f(\valconf(\overline{r})) \neq 0
}{
  (\pcconf, \valconf, \bufconf)
  \transition{(\athread, loc)}
  (\pcconf', \valconf, \bufconf)
}\text{~~local assertion}
\end{equation*}
\caption{Transition rules from state $(\texttt{pc}, \texttt{val}, \texttt{buf})$ with $\texttt{pc}(t) = l$ and instruction $l: instr$; goto $l'$. 
Throughout the definition, we let $\texttt{pc}' := \texttt{pc}[\athread := l']$.
The first three rules make sure that a read will get the value of the latest buffered store from the own thread, if any, and otherwise the value from memory. 
We call a read from one of the buffers an \emph{early read}.  
The \emph{issue store} rule puts a store into the thread's per-address buffer. 
The \emph{advance buffer} rule moves the oldest store from a per-address buffer to the all-addresses buffer. 
The \emph{store to memory} rule deletes the oldest store from the all-addresses buffer and reflects its effect in the global memory. 
The \emph{scfence} rule allows for the execution of the scfence command if the current thread's buffers are empty. 
The \emph{issue fence} rule allows for the execution of the fence command if the current thread's per-address buffers for the given addresses are empty. 
The \emph{fence} rule allows us to track the order of executed stores and fences.\label{fig:transitions-pso}
}
\end{figure}

The transition rules are shown in Figure~\ref{fig:transitions-pso}. 
They implement the behavior explained above. 
We use $\projection{-}{-}$ for the projection of a sequence to a subset of its actions. 
Note that the transitions are labeled by actions from  
\begin{align*}
\ACT := \THRD \times ((\{ \st, \ld \} \times \DOM^2 ) \cup \{ \isu, \loc, \scfence \} \cup (\{ \fence \} \times \DOM^*)).
\end{align*}
Every action contains its executing thread. 
An action can be a store or load annotated with an address and a value, an issue, local, or scfence action, or a fence action with a list of addresses. 
The $\isu$ action denotes the issue of a store or fence action.  We use $\isu_{\anaction}$ to refer to the issue that action $\anaction$ corresponds to.

A \emph{computation} is a sequence of actions that describes a single execution of a program.
The set of computations of program $\aprogram$ is 
\begin{align*}
\psocompof{\aprogram} := \{ \tau\in\ACT^* ~|~ s_0 &\transition{\tau} s = (\pcconf, \valconf, \bufconf) \text{ with }\\
&\forall \athread \in \THRD ~.~ \forall a \in \DOM ~.~ \bufconf(\athread) = \varepsilon \land \bufconf(\athread, a) = \varepsilon \}\ .
\end{align*}
It contains all sequences of actions following the transition relation starting at the initial state and ending at a state with empty buffers.
The latter requirement does not restrict the program behavior. 
Buffers can always be emptied.

\FloatBarrier

\section{Happens-before Relations}

The formal definition for the relations of happens-before is as follows:
The program-order relation of a thread $\athread$ is obtained by projecting computation $\tau$ to the actions of this thread.
For stores and fences that are subject to delays, what matters is the moment they are issued.
The actual actions are also projected away, denoted by \emph{not delayed}:
\begin{align*}
\po^\athread\ :=\ \projection{\tau}{(\athread, \text{not delayed}, *, *)}.
\end{align*}
The \emph{program-order relation} of the program is then the union of the program orders relations of all threads, $\po\ :=\ \bigcup_{\athread\in\aprogram}\po^{\athread}$.

The store order is defined similarly. 
We first obtain the store order relation of each address $a$ by projecting $\tau$ to $(*, \st, a, *)$.
Then we define the \emph{store order} of the program to be the union of all $\sto^{a}$. 

The \emph{source relation} is defined via the partial source function $\mathrm{src}_\tau$ that maps the load actions in $\tau$ to the store actions in $\tau$. 
The function returns the store that a load reads its value.
If the function is undefined, the load reads the initial value. 
For $\tau = \alpha \cdot \ld \cdot \beta$ with $\threadof{\ld} = \athread$ and $\addrof{\ld} = a$, we define
\begin{align*}
\mathrm{src}_\tau(\ld) := \left\{
\begin{array}{ll}
	x, & \text{if } \projection{\alpha}{\isu_{((\athread, \st, a, *) \text{ in } \beta)}} = \gamma \cdot \isu_x\\
	y,\phantom{aaa} & \text{if }  \projection{\alpha}{\isu_{((\athread, \st, a, *) \text{ in } \beta)}} = \varepsilon ~\land~ \projection{\alpha}{(*, \st, a, *)} = \gamma \cdot y.
	\end{array}\right.
\end{align*}
The first case looks for the latest suitable store of the same thread that is already issued but not yet stored.
In this case, the load will perform an early read.
If there is no such store, the second case looks for the latest suitable store on the global memory. 
With this, we have $\mathrm{src}_\tau(\ld)\src \ld$, provided 
$\mathrm{src}_\tau(\ld)$ is defined.

The \emph{conflict relation} is derived from $\src$ and $\sto$.
We have $\ld\cf\st$, if there is a store $\st'$ with $\st'\src\ld$ and $\st'\sto\st$.
Moreover, we have $\ld\cf \st$ if both access the same address, there is no $\st'\src\ld$ and no $\st''\sto\st$.

\section{Missing Proofs}

To prove the dichotomy result, we first need to show a technical lemma.
It gives information about the shape of the last action that has been overtaken by a store.

\begin{lemma}
\label{last-overtaken-action}
In a minimal violation $\tau = \tau_1 \cdot \isu_{\st} \cdot \tau_2 \cdot \anaction \cdot \tau_3 \cdot \st \cdot \tau_4\in \psocompof{\aprogram}$ with $\threadof{\st} = \athread = \threadof{\anaction}$ and $\projection{\tau_3}{\athread} = \varepsilon$, action $\anaction$ is one of the following:
\begin{itemize}
\item a load of an address different from $\addrof{\st}$
\item a store
\item a \texttt{fence} instruction that was issued before $\isu_{\st}$.
\end{itemize}
\end{lemma}
\begin{proof}
We know that $\anaction$ is not an \texttt{scfence} instruction since $\st$ overtakes it. If $\anaction$ was a local or issue action, we could place it after $\st$ without changing the trace: $\tau' := \tau_1 \cdot \isu_{\st} \cdot \tau_2 \cdot \tau_3 \cdot \st \cdot \anaction \cdot \tau_4$ also is a violation. Since $\delaysof{\tau'} < \delaysof{\tau}$ we have $\mmof{\tau'} < \mmof{\tau}$.
This contradicts minimality of $\tau$.

If $\anaction$ was a load of $\addrof{\st}$, we could as well place it after $\st$. 
We note that $\anaction$ performs an early read in $\tau$. 
Since $\projection{\tau_3}{\athread}$ is empty, $\anaction$ reads from the same store in $\tau$ and $\tau'$. Therefore, the trace is preserved in this case, too. Again, $\tau'$ has less delays than $\tau$, which contradicts minimality of $\tau$.

If $\anaction$ was a \texttt{fence} instruction that was issued within $\tau_2$, then it could not contain $\addrof{\st}$. 
Otherwise, the fence would have moved $\st$ to the all-addresses buffer and from there $\st$ would have hit the global memory before $\anaction$.
This is due to the FIFO behavior of the buffer. 
Hence, the fence does not contain $\addrof{\st}$.
This means we can place it after $\st$ to obtain the new computation $\tau'$. 
As the delays for the fence increase by the same amount as they decrease for the store, $\tau'$ has the same number of delays as $\tau$. 
But $\tau'$ saves one reordering of $\st$ over the fence.
Therefore, we get a contradiction to minimality of $\tau$. \qed
\end{proof}
\begin{proof}[of Lemma~\ref{dichotomy}]
The lemma is equivalent to $\neg (i)$ implies $(ii)$, which we prove by induction on the length of $\tau_2$.

{\bf Base case}\quad If $\lengthof{\tau_2}$ = 0, we know that $\tau = \tau_1 \cdot \anaction \cdot \anotheraction \cdot \tau_3$ and $\anaction \not\hb \anotheraction$. 
If $\threadof{\anaction} = \threadof{\anotheraction}$, then $\anotheraction \po^+ \anaction$. 
This means $\anotheraction$ is a store or a \texttt{fence}.  
Due to the overtake of $\anotheraction$, action $\anaction$ will not be an $\texttt{scfence}$.
If $\anaction$ was an issue, a local assignment, or an assert, then we could save the overtake of $\anotheraction$ over $\anaction$.
The resulting computation $\tau'$ would still be violating but have a smaller cost than $\tau$.
A contradiction to minimality of $\tau$. 
So $\anaction$ is a store, a load, or a $\texttt{fence}$. 
Assume $\anotheraction$ is a store. 
By FIFO, $\anaction$ is not a store to $\addrof{\anotheraction}$. 
Together with Lemma~\ref{last-overtaken-action}, we know that $\anaction$ is a load of an address or a store to an address different from $\addrof{\anotheraction}$. 
In both cases, we can swap $\anaction$ and $\anotheraction$ leading to a computation $\tau'$ that has the same trace as $\tau$ and therefore is a violation, too. 
But $\tau'$ saves one delay (if $\anaction$ is a load) or one reordering while preserving the delays (if $\anaction$ is a store). 
This contradicts minimality of $\tau$. 
Assume $\anotheraction$ is a $\texttt{fence}$. 
If $\anaction$ is a load, we can simply swap the commands.
This saves a delay and leads to a contradiction to minimality. 
If $\anaction$ is a store or a fence, it would have entered the all-addresses buffer after the fence $\anotheraction$ and have left it earlier.
This contradicts the FIFO order.

If $\threadof{\anaction} \neq \threadof{\anotheraction}$, then we know from $\anaction \not\hb \anotheraction$ that $\anaction$ and $\anotheraction$ are not two load or store actions operating on the same address where at least one is a store. 
Therefore, we can swap $\anaction$ and $\anotheraction$ leading to a computation $\tau'$ that has the same trace and per-thread-computations as $\tau$. 
Since $\tau_{22}$ is empty, there is nothing more to show.

{\bf Step case}\quad Assume the induction hypothesis holds for $\lengthof{\tau_2} \leq n$. We consider a computation $\tau = \tau_1 \cdot \anaction \cdot \tau_2 \cdot \anotheraction \cdot \tau_3$ with $\lengthof{\tau_2} = n + 1$.
Let $\tau_2 = \tau_2' \cdot \athirdaction$. 
Since $\anaction \not\hb^+ \anotheraction$ through $\tau_2$, we have $\anaction \not\hb^+ \athirdaction$ through $\tau_2'$ or $\athirdaction \not\hb \anotheraction$.

If $\anaction \not\hb^+ \athirdaction$, then we can apply the induction hypothesis to $\tau$ with respect to $\anaction$ and $\athirdaction$. 
This results in $\tau' = \tau_1 \cdot \tau_{21}' \cdot \athirdaction \cdot \anaction \cdot \tau_{22}' \cdot \anotheraction \cdot \tau_3$. 
Applying the hypothesis again to $\tau'$ with respect to $\anaction$ and $\anotheraction$ yields $\tau'' = \tau_1 \cdot \tau_{21}' \cdot \athirdaction \cdot \tau_{221}' \cdot \anotheraction \cdot \anaction \cdot \tau_{222}' \cdot \tau_3$. In both applications of the hypothesis the trace and the per-thread-computations are preserved. Furthermore, $\tau_{222}'$ is a subsequence of $\tau_{22}'$ which is a subsequence of $\tau_2'$. 
Sequence $\tau_2'$ in turn is a subsequence of $\tau_2$.
We conclude with the observation that the subsequence relation is transitive.

If $\athirdaction \not\hb \anotheraction$, then we can apply the induction hypothesis to $\tau$ with respect to $\anotheraction$ and $\athirdaction$. This results in $\tau' = \tau_1 \cdot \anaction \cdot \tau_2' \cdot \anotheraction \cdot \athirdaction \cdot \tau_3$. Applying the hypothesis again to $\tau'$ with respect to $\anaction$ and $\anotheraction$ gives $\tau'' = \tau_1 \cdot \tau_{21}' \cdot \anotheraction \cdot \anaction \cdot \tau_{22}' \cdot \athirdaction \cdot \tau_3$. As above, in both steps the trace and the per-thread-computations are preserved and $\tau_{22}'$ is a subsequence of $\tau_2'$ which is a subsequence of $\tau_2$.\qed
\end{proof}

As preparation to prove the locality result, we study the shape of the part of a computation that follows a cycle.
This rear part can only consist of delayed store and fence actions that respect program order.
\begin{lemma}[Rear]
\label{rear-part-stores}
Consider a minimal violation $\tau = \tau_1 \cdot \anaction \cdot \tau_2 \cdot \anotheraction \cdot \tau_3\in \psocompof{\aprogram}$ with $\anaction \hb^+ \anotheraction$ through $\tau_2$ and $\anotheraction \po^+ \anaction$. Then $\anotheraction\cdot\tau_3$ contains only delayed stores and fences and these actions respect program order. 
\end{lemma}
\begin{proof}
Towards a contradiction, assume $\anotheraction\cdot\tau_3$ contains further actions. 
We project $\anotheraction\cdot \tau_3$ to the stores and fences issued in $\tau_1 \cdot \anaction \cdot \tau_2$. 
The projection will keep $\anotheraction$ as it was issued before $\anaction$. 
Hence, the resulting computation $\tau'$ preserves the happens-before cycle between $\anaction$ and $\anotheraction$ and therefore is a violation, too. 
Furthermore, $\tau'$ is still a valid computation as we did not remove program-order-intermediate actions and there are no failing assertions after $\tau_2$.  Removing actions does not increase the delays or reorders, but it reduces the length. 
As $\mmof{\tau'} < \mmof{\tau}$, we have a contradiction to minimality of $\tau$.

If $\anotheraction \cdot \tau_3 = \alpha \cdot x \cdot \beta \cdot y \cdot \gamma$ with $y \po^+ x$, we can rearrange all actions to appear in program order. 
Bringing actions in program order will not violate FIFO or fence restrictions. It might change the store order but it does not affect the happens-before cycle between $\anaction$ and $\anotheraction$. 
If $\anotheraction$ is moved behind program-order-earlier stores and fences of the same thread in $\tau_3$, the cycle remains.
The reason is that the happens-before-through relation is stable under insertion of actions.  
Rearranging the stores and fences does not change the number of delays, but saves at least the reordering of $y$ over $x$.  This contradicts minimality of $\tau$. \qed
\end{proof}
The following normalization lemma shows that, given an action, one can remove all successors that are independent of this action.
To be more precise, we move the independent successors to the left of the given action. 
Technically, we apply  Lemma~\ref{dichotomy} to construct a new computation where the independent successors execute earlier.
\begin{lemma}[Dependence]
\label{hb-dependent-part}
Let $\tau = \alpha \cdot \anaction \cdot \beta \cdot \gamma\in \psocompof{\aprogram}$ be a minimal violation. There is a minimal violation $\tau' = \alpha \cdot \beta' \cdot \anaction \cdot \beta'' \cdot \gamma\in  \psocompof{\aprogram}$ where $\traceof{\tau'}=\traceof{\tau}$, $\projection{\tau'}{\athread}=\projection{\tau}{\athread}$ for all threads $\athread$, $\beta''$ a subsequence of $\beta$, and where
\begin{center}
for every action $\anotheraction$ in $\beta''$ we have $\anaction\hb^+\anotheraction$ through the part in between.
\end{center}
\end{lemma}
\begin{proof}
We proceed by induction on the length of $\beta$.
The base case of an empty sequence is trivial.
Assume the statement holds for $\lengthof{\beta} \leq n$ and consider a computation $\tau = \alpha \cdot \anaction \cdot \beta \cdot \gamma$ with $\lengthof{\beta} = n+1$. 
If every action in $\beta$ is happens-before-through dependent on $\anaction$, then $\tau' := \tau$ satisfies the requirements.  
Otherwise, we take the first action $\anotheraction$ in $\beta$ so that $\tau = \alpha \cdot \anaction \cdot \beta_1 \cdot \anotheraction \cdot \beta_2 \cdot \gamma$ and $\anaction \not\hb^+ \anotheraction$ through $\beta_1$. 
An application of Lemma~\ref{dichotomy} gives a computation $\tau' = \alpha \cdot \beta_{11} \cdot \anotheraction \cdot \anaction \cdot \beta_{12} \cdot \beta_2 \cdot \gamma$ with the same trace and per-thread-computations as $\tau$ and $\beta_{12}$ being a subsequence of $\beta_1$. Since $\lengthof{\beta_{12} \cdot \beta_2} < \lengthof{\beta}$ (we took out $\anotheraction$), we can now apply the induction hypothesis to $\tau'$. 
It gives another computation $\tau'' = \alpha \cdot \beta_{11} \cdot \anotheraction \cdot (\beta_{12} \cdot \beta_2)' \cdot \anaction \cdot (\beta_{12} \cdot \beta_2)'' \cdot \gamma$, where $(\beta_{12} \cdot \beta_2)''$ is a subsequence of $\beta_{12} \cdot \beta_2$. 
As $\beta_{12}$ is a subsequence of $\beta_1$ and $\beta_1 \cdot \beta_2$ is a subsequence of $\beta$, we conclude that $(\beta_{12} \cdot \beta_2)''$ is a subsequence of $\beta$. 
As Lemma~\ref{dichotomy} and the induction hypothesis preserve the per-thread-computations and the trace, the corresponding equalities are fulfilled, too. 
The dependency requirement directly follows from the induction hypothesis.\qed
\end{proof}

\subsection{Locality}
\begin{proof}[of Theorem~\ref{locality}]

If the program admits violations, then there is one that is local in the sense that only one thread delays actions.
Towards a contradiction, we assume that in all violations at least two threads delay actions.
(If no thread delays actions, the computation is SC.)
Among these violations, let $\tau$ be a minimal one.
We define:
\begin{align*}
y&:= \text{last $\isu$ or load action that has been overtaken by a store}\\
\st_y&:= \text{earliest store that overtook $y$}\\
\athread_y&:= \threadof{y} (= \threadof{\st_y})\\
x &:= \text{last $\isu$ or load from a thread $\athread_x \neq \athread_y$ that has been overtaken by a store}\\
\st_x&:= \text{earliest store that overtook $x$}.
\end{align*}
\begin{lemma}\label{existence}
$x$, $y$, $\st_x$, and $\st_y$ exist.
\end{lemma}
We defer the proof of the lemma for the moment.
Given $y$, $x$, $\st_y$, and $\st_x$, the computation takes one of the following forms (for some $\tau_1,\ldots, \tau_5\in\ACT^*$):
\begin{align}
\tau &= \tau_1 \cdot x \cdot \tau_2 \cdot \st_x \cdot \tau_3 \cdot y \cdot \tau_4 \cdot \st_y \cdot \tau_5\label{forma}\\
\tau &= \tau_1 \cdot x \cdot \tau_2 \cdot y \cdot \tau_3 \cdot \st_y \cdot \tau_4 \cdot \st_x \cdot \tau_5\label{formb}\\
\tau &= \tau_1 \cdot x \cdot \tau_2 \cdot y \cdot \tau_3 \cdot \st_x \cdot \tau_4 \cdot \st_y \cdot \tau_5\ .\label{formc}
\end{align}
The task is to show that each of the forms yields a contradiction.\\[0.2cm]
{\bf Form~\eqref{forma}:}\quad 
We apply Proposition~\ref{hb-cycle-existence} to $x$ and $\st_x$ and get a happens-before cycle within $\tau_1 \cdot x \cdot \tau_2 \cdot \st_x$. 
By Lemma~\ref{rear-part-stores}, the part starting with $\st_x$ consists of stores and fences only.
This contradicts the fact that $y$ is an $\isu$ or a load.\\[0.2cm]
{\bf Form~\eqref{formc}:}\quad
With Proposition~\ref{hb-cycle-existence}, we get a happens-before cycle within $\tau_1 \cdot x \cdot \tau_2 \cdot y \cdot \tau_3 \cdot \st_x$. 
By Lemma~\ref{rear-part-stores}, the part after $\st_x$ consists of delayed stores and fences only. 
Proposition~\ref{local-hb-dependence} yields a happens-before chain from $x$ to $\st_x$. 
We can rearrange the computation as follows: 
\begin{align*}
\tau' = \tau_1 \cdot x \cdot \tau_2 \cdot y \cdot \tau_3 \cdot (\projection{\tau_4}{\athread_y}) \cdot \st_y \cdot \st_x \cdot (\projection{\tau_4}{\text{rest}}) \cdot \tau_5\ .
\end{align*} 
The happens-before chain is still present as the relation is stable under insertion. 
As we did not change the per-thread-computations, $\tau'$ is a minimal violation, too. 
It has Form~\eqref{formb}, which yields a contradiction as we will show next.\\[0.2cm]
{\bf Form~\eqref{formb}:}\quad
We will show that there are two happens-before cycles, one between $x$ and $\st_x$, the other between $y$ and $\st_y$. 
The ultimate goal is to move $\st_y$ back over $y$ to save one delay. 
We can do this in such a way that there remains a happens-before cycle between $x$ and $\st_x$. 
The resulting computation is therefore a violation with less delays, which contradicts minimality of $\tau$.

By Corollary~\ref{local-hb-dependence}, we have $y\hb^+ \st_y$ through $\tau_3$. 
With this, Lemma~\ref{rear-part-stores} applies and shows that $\tau_4\cdot \st_x\cdot \tau_5$ only consists of delayed stores and fences.

We apply Lemma~\ref{hb-dependent-part} to $\tau_3$ and obtain 
\begin{align*}
\tau' := \tau_1 \cdot x \cdot \tau_2 \cdot \tau_3' \cdot y \cdot \tau_3'' \cdot \st_y \cdot \tau_4 \cdot \st_x \cdot \tau_5.
\end{align*}
Computation $\tau'$ is again a minimal violation.
It is a violation as we do not change the trace.
It is minimal as we do not change the per-thread computations, and the delays and reorders are counted per thread.  

By Lemma~\ref{hb-dependent-part}, every action in $\tau_3''$ is happens-before-through dependent on action $y$. 
By Proposition~\ref{local-hb-dependence}, we moreover have $y \hb^+ \st_y$ through $\tau_3''$. 
This means, we find an access to $\addrof{\st_y}$ in $\tau_3''$. 
Let $\acc$ be the first such access and let $\tau_3'' = \tau_{3a}'' \cdot \acc \cdot \tau_{3b}''$. 
We move $\st_y$ over $\tau_{3b}''$.
Since the store cannot be moved over fences due to the FIFO all-addresses buffer, we obtain the computation
\begin{align*}
\tau'' := \tau_1 \cdot x \cdot \tau_2 \cdot \tau_3' \cdot y \cdot \tau_{3a}'' \cdot \acc \cdot 
(\projection{\tau_{3b}''}{&\text{fences of }\athread_y}) \cdot \st_y \cdot\\ 
&(\projection{\tau_{3b}''}{\text{delayed}}) \cdot \tau_4 \cdot \st_x \cdot \tau_5\ .
\end{align*}
Here, $\projection{\tau_{3b}''}{\text{delayed}}$ is the projection of $\tau_{3b}''$ to the delayed stores and fences (of threads different from $\athread_y$). 
Computation $\tau''$ is executable, i.e., in the set of computations of the program. 
Indeed, starting from $\acc$ it does not contain any assertions that could fail. 
The assertions before $\acc$ continue to hold as $\tau'$ was executable. 
Computation $\tau''$ is a violation due to the cycle $y \hb^+ \acc \cfst \st_y \po^+ y$.
It is minimal as we do not change the computation of thread $\athread_y$ and potentially only remove delays or reorders from the computations of the other threads. 
Actually, there is no removal. 
As we started from a minimal violation, we get that
$\projection{\tau_{3b}''}{\text{delayed}}$ has to be all of $\tau_{3b}''$ except for the fences of $\athread_y$.

Applying Corollary~\ref{local-hb-dependence}, we get $x$ happens before $\st_x$ through
\begin{align}
\tau_2 \cdot \tau_3' \cdot y \cdot \tau_{3a}'' \cdot acc \cdot 
(\projection{\tau_{3b}''}{\text{fences of }\athread_y}) \cdot \st_y \cdot
 (\projection{\tau_{3b}''}{\text{delayed}}) \cdot \tau_4\ .\label{through}
\end{align}
Our goal is to move $\st_y$ before $y$ in order to save a delay.
This move should preserve the relation $x\hb^+\st_x$ through~\eqref{through}. 
Moving $\st_y$ also moves all fences of $\athread_y$ in $\tau_{3a}''$.
Since $(\projection{\tau_{3a}''}{\text{fences of }\athread_y})\cdot (\projection{\tau_{3b}''}{\text{fences of }\athread_y})=(\projection{\tau_{3}''}{\text{fences of }\athread_y})$, the resulting computation can be written as
\begin{align*}
\tau''':=\tau_1 \cdot x \cdot \tau_2 \cdot \tau_3'\cdot
(\projection{\tau_3''}{&\text{fences of }\athread_y}) \cdot \st_y \cdot y \cdot \\
&(\projection{\tau_{3a}''}{\text{rest}}) \cdot acc \cdot 
(\projection{\tau_{3b}''}{\text{delayed}}) \cdot \tau_4 \cdot \st_x \cdot \tau_5\ .
\end{align*}
The computation is executable because $\projection{\tau_{3a}''}{\text{rest}}$ does not contain an access to $\addrof{\st_y}$.
This is due to the choice of $\acc$ as the first access to the address.

It remains to show that $x \hb^+ \st_x$ through~\eqref{through} still holds. 
If $\st_y$ is not on this happens-before-through chain, we are done. 
If the chain contains $\st_y$, we have that $x\hb^+\anaction \popluscfst \st_y\sto\st \sto^* \st_x$ for some actions $\anaction$ and $\st$ in~\eqref{through}$\cdot \st_x$ ($\st$ can be $\st_x$). 
If we have a program-order dependence from $\anaction$ to $\st_y$, this is not changed by moving $\st_y$ back over $y$. 
Computation $\tau'''$ contains the chain $x\hb^+\anaction \po^+ \st_y \srcst \acc \cfst \st \sto^* \st_x$. 
Indeed, as $y$ is chosen to be a load or an issue, we know that action $\anaction$ (which is program-order earlier than $\st_y$) has to lie before $y$ in~\eqref{through}. 
If we have a conflict or a store dependence from $\anaction$ to $\st_y$, we know that $\anaction = \acc$.
In this case, $\tau'''$ contains the chain $x \hb^+ \acc \cfst \st \sto^* \st_x$. 
Therefore, $\tau'''$ contains the cycle $x \hb^+ \st_x \po^+ x$ and thus is a violation. 
As the rearrangement saves one delay of $y$, $\tau'''$ contradicts minimality of $\tau''$.\qed	
\end{proof}

\begin{proof}[of Lemma~\ref{existence}]
To argue for the existence, we first show that each thread that delays actions will in particular delay stores.
Assume this is not the case and the thread only delays fences.
Take the first fence that gets delayed. 
It is not delayed over stores. 
The FIFO property of the all-addresses buffer would force the stores to be delayed with the fence, but there are no delayed stores. 
Furthermore, the fence cannot be delayed over other fences. 
There is no earlier delayed fence, as we took the first one, and there cannot be a non-delayed fence due to the FIFO property that forbids reorderings among fences.
The fence is thus delayed over loads, local actions, and issues only. Overtakes of fences over these actions can be dropped without changing the trace while reducing the cost of the computation. A contradiction to minimality of $\tau$.

To find the required $y$ and $\st_y$ and $x$ with $\st_x$, we project $\tau$ to a thread $\athread$ that delays actions.
With the previous reasoning, the thread delays at least one store $\st$. 
Let the last action overtaken by $\st$ be $\anaction$.
We thus have
\begin{align*}
\projection{\tau}{\athread} = \tau_1 \cdot \isu_\st \cdot \tau_2 \cdot a \cdot \st \cdot \tau_3\quad\text{for some }\tau_1, \tau_2, \tau_3\in \ACT^*.
\end{align*} 
By Lemma~\ref{last-overtaken-action}, $a$ can be a load, a store, or a fence. 
In case of a load or a store issued in $\tau_2$, we are done. 
If $a$ is a store issued in $\tau_1 = \tau_{1a} \cdot \isu_a \cdot \tau_{1b}$, then $\isu_\st$ is an issue action overtaken by the store $a$. 
If $a$ is a fence, the situation is more delicate.
In this case, the full computation is 
\begin{align*}
\tau = \tau_{1} \cdot \isu_\fence \cdot \tau_{2} \cdot \isu_\st \cdot \tau_3 \cdot \fence \cdot \tau_4\cdot \st \cdot \tau_5\quad\text{for some }\tau_1,\ldots, \tau_5\in \ACT^*.
\end{align*} 
An application of Corollary~\ref{local-hb-dependence} gives $\isu_{\st} \hb^+ \fence$ through $\tau_3$. 
As there is no forward $\po^+$ relation from $\isu_{\st}$ to $\fence$, we will find actions of thread $\athread$ on this path. 
The $\isu_{\st}$ action can only have an outgoing $\po$ dependence to a local, a load, or an issue action in $\tau_3$.
The load or the issue would serve as the required action overtaken by $\st$.
For the local action, the same reasoning applies.
Since we eventually need a $\cf$ or $\sto$ dependence to form a happens-before-through chain, and since local actions have no such outgoing dependencies, we are sure to find a load or an issue.\qed
\end{proof}

\begin{proof}[of Proposition~\ref{witnesses}]
Let $\athread_A$ be the one thread delaying actions. 
The first action that is delayed has to be a store. 
If it was a fence, we could simply save the delay by executing it directly after its issue action. 
There will be no store with a fenced address in between. So let $\st$ be the first action that is delayed and let $\anaction$ be the last action that is overtaken by $\st$. We have 
\begin{align*}
\tau = \tau_1 \cdot \isu_\st \cdot \tau_2 \cdot a \cdot \tau_3 \cdot \st \cdot \tau_4
\end{align*} with $\projection{\tau_3}{\athread_A} = \varepsilon$. 
By Lemma~\ref{last-overtaken-action}, $\anaction$ is a load, a store, or a fence action that was issued in $\tau_1$. 
The latter one is not possible as $\st$ is the first action that was delayed. 
Therefore, $\anaction$ is either a load or a store issued in $\tau_2$ and we have $\st \po^+ \anaction$. 
With Corollary~\ref{local-hb-dependence}, we get $\anaction \hb^+ \st$ through $\tau_3$. 
Altogether, there is the happens-before cycle $\st \po^+ \anaction \hb^+ \st$. 
Lemma~\ref{rear-part-stores} ensures that $\tau_4$ consists of landing stores and fences only. 
As $\athread_A$ is the only thread delaying actions, $\tau_4$ solely contains actions of thread $\athread_A$.

We can show that $\tau_2 \cdot \anaction$ does not contain delayed actions. 
Towards a contradiction, assume there is a first action $x$ in $\tau_2 \cdot \anaction$ that is delayed over the last action $\anotheraction$. 
As $\st$ is the first action that gets delayed, $x$ is issued in $\tau_2$.
We have $\tau_2 \cdot \anaction = \alpha \cdot \isu_x \cdot \beta \cdot \anotheraction \cdot \gamma \cdot x \cdot \delta$ with $\threadof{\anotheraction} = \threadof{x} = \athread_A$ and $\projection{\gamma}{\athread_A} = \varepsilon$.
We know that $x$ has to be a store. 
If it was a delayed fence, we argue that we could drop a delay of $x$. 
If $b$ was a store or fence action, then it was not delayed as $x$ is the first delayed action. 
This implies a reordering between $x$ and $b$. 
Reorderings between fences as well as reordering a later issued store before an earlier issued fence is not possible due to the FIFO property of the all-addresses buffer.
In the remaining cases, $b$ being an issue, local, or load action, we can drop the delay of $x$ over $b$, which again contradicts minimality of $\tau$.

By Lemma~\ref{last-overtaken-action}, $\anotheraction$ is either a load, a store, or a fence issued before $\isu_x$. 
The latter fence is not possible as the only delayed action earlier than $x$ is $\st$. 
With the same argument, we know that $\anotheraction$ will be issued in $\beta$ if it is a store.
Corollary~\ref{local-hb-dependence} yields the cycle $\anotheraction \hb^+ x \po^+\anotheraction$. 
Applying Lemma~\ref{rear-part-stores}, we know that the actions starting with $x$ have to respect program order, in particular $x \po^+ \st$. But $x$ is issued after $\isu_\st$, which contradicts that order.

Furthermore, we can assume that for every action $\athirdaction$ in $\tau_3\cdot \st$ we have $\anaction \hb^+ \athirdaction$. 
If this is not the case, we can apply Lemma~\ref{hb-dependent-part} to get a minimal violation that fulfills that property.\qed
\end{proof}

\subsection{Singularity}

\begin{proof}[of Proposition~\ref{no-move-of-two-stores}]
Towards a contradiction, assume there is such a computation and let $\athread:=\threadof{\st_1} = \threadof{\st_2}$. 
We have $\st_1 \po^+ \st_2$, for otherwise we could swap the stores violating minimality of $\tau$. 
Action $\st_2$ is issued in $\tau_1$, say $\tau_1 = \tau_{1a} \cdot \isu_{\st_2} \cdot \tau_{1b}$. 
The overall strategy will be to move $\st_1$ back over $\isu_{\st_2}$.

We find the first action $\acc'$ that is
(i) before $\st_1$ in the relation $\cfst\po^*$ and (ii) placed on an $\isu_{\st_2}\hb^+\st_1$ chain through $\tau_{1b}$.
Intuitively, action $\acc'$ is from another thread ($\cfst$) and accesses an address that leads to a program-order dependence ($\po^*$).   
Technically, we first apply Lemma~\ref{hb-dependent-part} to $\tau_{1b}$ and get
\begin{align*}
\tau' = \tau_{1a} \cdot \tau_{1b}' \cdot \isu_{\st_2} \cdot \tau_{1b}'' \cdot \st_1 \cdot \st_2 \cdot \tau_2\ .
\end{align*} 
By Corollary~\ref{local-hb-dependence}, we have $\isu_{\st_2} \hb^+ \st_1$ through $\tau_{1b}''$. 
This path contains an access $\acc'$ with  
(1) $\isu_{\st_2} \hb^+ \acc' \cfst \st_1$ or 
(2) $\isu_{\st_2} \hb^+ \acc' \cfst \st' \po^+ \st_1$.  
Looking at the second case, $\st'$ really has to be a store since there are only stores in $\tau_{1b}''$ that are program-order-earlier than $\st_1$.

We are interested in the first access to $\addrof{\acc'}$.
To be precise, we consider all possible choices for $\acc'$. 
For each $\acc'$, we find the first access $\acc$ to $\addrof{\acc'}$ in $\tau_{1b}'' = \tau_{1b1}'' \cdot \acc \cdot \tau_{1b2}''$. 
Among all possible $\acc$, we take the first one.

We place all interesting actions together: 
\begin{align*}
\tau'' := \tau_{1a} \cdot \tau_{1b}' \cdot \isu_{\st_2} \cdot \tau_{1b1}'' \cdot \acc \cdot (\projection{\tau_{1b2}''}{\athread}) \cdot \st_1 \cdot \st_2 \cdot (\projection{\tau_{1b2}''}{\text{rest}}) \cdot \tau_2\ .
\end{align*} 
Assume we projected the part after $\acc$ to landing stores. 
Then it is easy to see that $\tau''$ is executable and has the same per-thread-computations as $\tau'$ or even less delays.
We have $\isu_{\st_2} \hb^+ \acc$, (1) $\acc \cfst \st_1$ respectively (2) $\acc \cfst \st' \po^+ \st_1$, and $\st_1 \po^+ \st_2$. 
Because of this happens-before cycle, $\tau''$ is a minimal violation. 
Minimality of $\tau''$ ensures that the assumed projection was not necessary, since the part in question did not contain any non-store action. 
The stores after $\acc$ are consistent with the program order as rearranging the stores would keep the cycle and save some reorders.

To conclude the argumentation, we rely on the following lemma, the proof of which we defer for the moment.
\begin{lemma}\label{allstores}
All stores in $(\projection{\tau_{1b2}''}{\athread}) \cdot \st_1 \cdot \st_2$ are to the address $\addrof{\acc}$.
\end{lemma}
With this result, we can move $\st_1$ before $\isu_{\st_2}$while preserving the cycle:
\begin{align*}
\tau''':=\tau_{1a} \cdot \tau_{1b}' \cdot (\projection{\tau_{1b2}''}{\athread}) \cdot \st_1 \cdot \isu_{\st_2} \cdot \tau_{1b1}'' \cdot \acc \cdot \st_2 \cdot (\projection{\tau_{1b2}''}{\text{rest}}) \cdot \tau_2\ . 
\end{align*}
Since there is no assertion after $\acc$ and $\acc$ is the only action that could have loaded a different value, the computation is executable. 
It still contains the cycle $\isu_{\st_2}\hb^+\acc\cfst\st_2$. 
A contradiction to minimality of $\tau''$.\qed
\end{proof}
\begin{proof}[of Lemma~\ref{allstores}]
We look at $(\projection{\tau_{1b2}''}{\athread}) \cdot \st_1$ first. 
Assume this part contains a store to a different address and take the first one.  
By the choice of $\acc$, $\tau_{1b1}''$ does not contain an access by another thread to that address. 
We can place the store directly before $\isu_{\st_2}$. 
The resulting computation is still executable, has the same trace, and saves one delay (over $\isu_{\st_2}$). 
This contradicts minimality of $\tau''$.

Now we look at $\st_2$. 
Assume it is to an address different from $\addrof{\acc}$. If there is a first access $\acc_2$ of another thread to $\addrof{\st_2}$ in $\tau_{1b1}'' = \tau_{1b1a}'' \cdot \acc_2 \cdot \tau_{1b1b}''$, we can use this access to have a cycle in a rearranged computation with less delays: $\tau_{1a} \cdot \tau_{1b}' \cdot (\projection{\tau_{1b2}''}{\athread}) \cdot \st_1 \cdot \isu_{\st_2} \cdot \tau_{1b1a}'' \cdot \acc_2 \cdot \st_2 \cdot \ldots$ (append the remaining actions in the order they appear in $\tau''$ and project them to landing stores). 
Here we have $\isu_{\st_2} \hb^+ \acc_2 \cfst \st_2$.
So there is no such access $\acc_2$.

We can place action $\st_2$ directly after $\isu_{\st_2}$ without changing the trace. 
If $\tau_{1b1}''$ contains an action of $\athread$ that is not a delayed store, this already contradicts minimality of $\tau''$ as we save delays. 
We show that $\isu_{\st_2} \hb^+ \acc$ through $\tau_{1b1}''$ implies the existence of such a not-delayed-store. 
The path begins with $\isu_{\st_2}$ and has to switch to $\acc$ somewhere.
The outgoing relation of an issue can only be program order. 
Only load and store actions create conflict or store dependencies to $\acc$.
Because of the program order, such a store would not be delayed.  \qed
\end{proof}

\section{Details of the Instrumentation}
\subsection{Instrumentation of the Attacker for Locality}
To begin with, we give the instrumentation of the attacker.
For a precise definition, let $\attacker$ declare registers $\areg^*$,
have initial location $\alab_{0}$, and define instructions
$\langle linst\rangle^*$ that contain $\stinst$ and $\lastinst$ from the attack.
The instrumentation is
\begin{align*}
\sem{\attacker}:= \lit*{thread}\ \tilde\attacker\ & \lit*{regs}\ \areg^*, \auxaddress, \auxfence\ \lit*{init}\ \alab_0\\
&\lit*{begin}\ \langle linst\rangle^*\ \semattackerstinst{\stinst}\ \semattackerlastinst{\lastinst}\ \semattacker{\langle linst\rangle}^* \ \attackerwait\ \lit*{end}\ ,
\end{align*}
where the functions are defined in Figure~\ref{Figure:TranslationAttackerLocality}.

\begin{figure}[ht]
	\begin{eqnarray}
	\semattackerstinst{\thetransition{\alab_1}{\alab_2}{\thestore{\anexpr_1}{\anexpr_2}}} &:=&\thetransition{\alab_1}{\tilde\alab_{x}}{\thestore{\auxdelayed{\anexpr_1}}{\auxdelayed{\anexpr_2}}}\label{Equation:StoreInst}\\
	&&\thetransition{\tilde\alab_{x}}{\tilde\alab_2}{\thelocal{\auxaddress}{\anexpr_1}}\notag\\[1mm]
	%
	\semattackerlastinst{\thetransition{\alab_1}{\alab_2}{\theload{\areg}{\anexpr}}}&:=&
	\thetransition{\tilde\alab_{1}}{\tilde\alab_{x1}}{\thecondition{\themem{\auxdelayed{\anexpr}} = 0}}\label{Equation:LastInstLoad}\\
	&&\thetransition{\tilde\alab_{x1}}{\tilde\alab_{x2}}{\thestore{\auxhb}{\mytrue}}\notag\\
	&&\thetransition{\tilde\alab_{x2}}{\attackerwaitlabel}{\thestore{\auxaccesslevel{\anexpr}}{\loadacc}}\notag\\[1mm]
	%
	\semattackerlastinst{\thetransition{\alab_1}{\alab_2}{\thestore{\anexpr_1}{\anexpr_2}}}&:=&
	\thetransition{\tilde\alab_1}{\tilde\alab_{x1}}{\thecondition{\auxfence = 0}}\label{Equation:LastInstStore}\\
	&&\thetransition{\tilde\alab_{x1}}{\tilde\alab_{x2}}{\thecondition{\themem{\auxdelayed{\anexpr_1}} = 0}}\notag\\
	&&\thetransition{\tilde\alab_{x2}}{\tilde\alab_{x3}}{\thestore{\anexpr_1}{\anexpr_2}}\notag\\
	&&\thetransition{\tilde\alab_{x3}}{\tilde\alab_{x4}}{\thestore{\auxhb}{\mytrue}}\notag\\
	&&\thetransition{\tilde\alab_{x4}}{\attackerwaitlabel}{\thestore{\auxaccesslevel{\anexpr_1}}{\storeacc}}\notag\\[1mm]
	%
	\semattacker{\thetransition{\alab_1}{\alab_2}{\thestore{\anexpr_1}{\anexpr_2}}} &:=&
	\thetransition{\tilde\alab_1}{\tilde\alab_{x1}}{\thecondition{\auxfence = 0}}\label{Equation:Store}\\
	&&\thetransition{\tilde\alab_{x1}}{\tilde\alab_{x2}}{\thecondition{\themem{\auxdelayed{\anexpr_1}} = 0}}\notag\\
	&&\thetransition{\tilde\alab_{x2}}{\tilde\alab_2}{\thestore{\anexpr_1}{\anexpr_2}}\notag\\
	&&\thetransition{\tilde\alab_1}{\tilde\alab_2}{\thestore{\auxdelayed{\anexpr_1}}{\auxdelayed{\anexpr_2}}}\notag\\[1mm]
	%
	\semattacker{\thetransition{\alab_1}{\alab_2}{\theload{\areg}{\anexpr}}} &:=&
	\thetransition{\tilde\alab_{1}}{\tilde\alab_{x1}}{\thecondition{\themem{\auxdelayed{\anexpr}} = 0}}\label{Equation:Load}\\
	&&\thetransition{\tilde\alab_{x1}}{\tilde\alab_{2}}{\theload{\areg}{\anexpr}}\notag\\
	&&\thetransition{\tilde\alab_{1}}{\tilde\alab_{x2}}{\thecondition{\themem{\auxdelayed{\anexpr}} \neq 0}}\notag\\
	&&\thetransition{\tilde\alab_{x2}}{\tilde\alab_{2}}{\theload{\auxdelayed{\areg}}{\auxdelayed{\anexpr}}}\notag\\[1mm]
	%
	\semattacker{\thetransition{\alab_1}{\alab_2}{\mathit{local}}}&:=& \thetransition{\tilde\alab_1}{\tilde\alab_2}{\mathit{local}}\label{Equation:Local}\\[1mm]
	%
	\semattacker{\thetransition{\alab_1}{\alab_2}{\thescfence}} &:=&\label{Equation:SCFence}\\[1mm]
	%
	\semattacker{\thetransition{\alab_1}{\alab_2}{\thefence{\anaddr_1 \ldots \anaddr_n}}} &:=&
	\thetransition{\tilde\alab_1}{\tilde\alab_2}{\areg_{fence} \leftarrow {\mytrue}}\label{Equation:Fence}\\
	&&\thetransition{\tilde\alab_1}{\tilde\alab_{x2}}{\thecondition{\themem{\auxdelayed{\anaddr_1}} = 0}}\notag\\
	&&\thetransition{\tilde\alab_{x2}}{\tilde\alab_{x3}}{\thecondition{\themem{\auxdelayed{\anaddr_2}} = 0}}\notag\\
	&&\vdots\notag\\
	&&\thetransition{\tilde\alab_{x(n-1)}}{\tilde\alab_{xn}}{\thecondition{\themem{\auxdelayed{\anaddr_{n-1}}} = 0}}\notag\\
	&&\thetransition{\tilde\alab_{xn}}{\tilde\alab_2}{\thecondition{\themem{\auxdelayed{\anaddr_n}} = 0}}\notag\\[1mm]
	%
	wait &:=&
	\thetransition{\attackerwaitlabel}{\tilde\alab_{x1}}{\thecondition{\themem{\auxaccesslevel{\auxaddress}} \neq 0}}\label{Equation:Wait}\\
	&&\thetransition{\tilde\alab_{x1}}{\tilde\alab_{x2}}{\thestore{\auxsuc}{\mytrue}}\notag
	\end{eqnarray}
	\caption{Instrumentation of the attacker for locality.}
	\label{Figure:TranslationAttackerLocality}
\end{figure}

\subsection{Optimized Instrumentation of the Attacker for Singularity}
The optimized translation for singularity is
\begin{align*}
\sem{\attacker}:= \lit*{thread}\ \tilde\attacker\ & \lit*{regs}\ \areg^*, \auxaddress, \auxdelayval\ \lit*{init}\ \alab_0\\
&\lit*{begin}\ \langle linst\rangle^*\ \semattackerstinst{\stinst}\ \semattackerlastinst{\lastinst}\ \semattacker{\langle linst\rangle}^* \ \attackerwait\ \lit*{end}\ ,
\end{align*}
where the functions are defined in Figure~\ref{Figure:OptimizedTranslationAttacker}.

\begin{figure}[ht]
	\begin{eqnarray}
	\semattackerstinst{\thetransition{\alab_1}{\alab_2}{\thestore{\anexpr_1}{\anexpr_2}}} &:=&\thetransition{\alab_1}{\tilde\alab_{x}}{\thelocal{\auxdelayval}{\anexpr_2}}\label{Equation:OptimizedStoreInst}\\
	&&\thetransition{\tilde\alab_{x}}{\tilde\alab_2}{\thelocal{\auxaddress}{\anexpr_1}}\notag\\[1mm]
	%
	\semattackerlastinst{\thetransition{\alab_1}{\alab_2}{\theload{\areg}{\anexpr}}}&:=&
	\thetransition{\tilde\alab_{1}}{\tilde\alab_{x1}}{\thecondition{\auxaddress \neq \anexpr}}\label{Equation:OptimizedLastInstLoad}\\
	&&\thetransition{\tilde\alab_{x1}}{\tilde\alab_{x2}}{\thestore{\auxhb}{\mytrue}}\notag\\
	&&\thetransition{\tilde\alab_{x2}}{\attackerwaitlabel}{\thestore{\auxaccesslevel{\anexpr}}{\loadacc}}\notag\\[1mm]
	%
	\semattackerlastinst{\thetransition{\alab_1}{\alab_2}{\thestore{\anexpr_1}{\anexpr_2}}}&:=&
	\thetransition{\tilde\alab_1}{\tilde\alab_{x1}}{\thecondition{\auxaddress \neq \anexpr_1}}\label{Equation:OptimizedLastInstStore}\\
	&&\thetransition{\tilde\alab_{x1}}{\tilde\alab_{x2}}{\thestore{\anexpr_1}{\anexpr_2}}\notag\\
	&&\thetransition{\tilde\alab_{x2}}{\tilde\alab_{x3}}{\thestore{\auxhb}{\mytrue}}\notag\\
	&&\thetransition{\tilde\alab_{x3}}{\attackerwaitlabel}{\thestore{\auxaccesslevel{\anexpr_1}}{\storeacc}}\notag\\[1mm]
	%
	\semattacker{\thetransition{\alab_1}{\alab_2}{\thestore{\anexpr_1}{\anexpr_2}}} &:=&
	\thetransition{\tilde\alab_1}{\tilde\alab_x}{\thecondition{\auxaddress \neq \anexpr_1}}\label{Equation:OptimizedStore}\\
	&&\thetransition{\tilde\alab_x}{\tilde\alab_2}{\thestore{\anexpr_1}{\anexpr_2}}\notag\\[1mm]
	%
	\semattacker{\thetransition{\alab_1}{\alab_2}{\theload{\areg}{\anexpr}}} &:=&
	\thetransition{\tilde\alab_{1}}{\tilde\alab_{x1}}{\thecondition{\auxaddress \neq \anexpr}}\label{Equation:OptimizedLoad}\\
	&&\thetransition{\tilde\alab_{x1}}{\tilde\alab_{2}}{\theload{\areg}{\anexpr}}\notag\\
	&&\thetransition{\tilde\alab_{1}}{\tilde\alab_{x2}}{\thecondition{\auxaddress = \anexpr}}\notag\\
	&&\thetransition{\tilde\alab_{x2}}{\tilde\alab_{2}}{\thelocal{\areg}{\auxdelayval}}\notag\\[1mm]
	%
	\semattacker{\thetransition{\alab_1}{\alab_2}{\mathit{local}}}&:=& \thetransition{\tilde\alab_1}{\tilde\alab_2}{\mathit{local}}\label{Equation:Local}\\[1mm]
	%
	\semattacker{\thetransition{\alab_1}{\alab_2}{\thescfence}} &:=&\label{Equation:SCFence}\\[1mm]
	%
	wait &:=&
	\thetransition{\attackerwaitlabel}{\tilde\alab_{x1}}{\thecondition{\themem{\auxaccesslevel{\auxaddress}} \neq 0}}\label{Equation:Wait}\\
	&&\thetransition{\tilde\alab_{x1}}{\tilde\alab_{x2}}{\thestore{\auxsuc}{\mytrue}}\notag
	\end{eqnarray}
	\caption{Optimized instrumentation of the attacker for singularity.}
	\label{Figure:OptimizedTranslationAttacker}
\end{figure}

\subsection{Instrumentation of Helpers}
The instrumentation of helpers is taken from~\cite{BDM13}.
We recall it here for completeness. 
Initially, a helper executes its actions normally. 
When the $\auxhb$ flag is set by the attacker, the helper executes only actions that are hb-dependent on $a$. 
To decide whether an action is hb-dependent on $a$, we use the following observation from~\cite{BDM13}. 
We have to track the maximal so-called \emph{access level} for each address, in the order
\begin{center}
\emph{no acces} $<$ \emph{load access} $<$ \emph{store access}\ .
\end{center}
So if an address $y$ has seen a load but no store, we will keep the value \emph{load access} in $(y, \auxhb)$. 
If the address has also seen a store, we keep \emph{store access} in $(y, \auxhb)$, even if there have been later loads.
The point is that the store is guaranteed to construct a dependency with further actions, while a load will only construct a conflict with subsequent stores.

For the formal definition, let the helper thread $\athread$ declare registers $\areg^*$, have initial label $\alab_0$, and define instructions $\langle linst\rangle^*$.
The instrumented thread is
\begin{align*}
\sem{\athread}:=\lit*{thread}\ \tilde\athread\ & \lit*{regs}\ \tilde\areg, \areg^*\ \lit*{init}\ \alab_0\\
&\lit*{begin}\ \semhelperorig{\langle linst \rangle}^*\  \semhelpertrans{\langle ldstinst\rangle}^*\ \semhelpercpy{\langle linst \rangle}^*\ \lit*{end}\ ,
\end{align*}
where the functions are defined in Figure~\ref{Figure:HelpersInstrumentation}.
\begin{figure}[t]
	\begin{eqnarray}
	\semhelperorig{\thetransition{\alab_1}{\alab_2}{\mathit{instr}}}&:=&
	\thetransition{\alab_1}{\alab_{x}}{\thecondition{\themem{\auxhb} = 0}}\label{Equation:HelperOriginal}\\
	&&\thetransition{\alab_x}{\alab_2}{\mathit{instr}}\notag\\
	%
	\semhelpertrans{\thetransition{\alab_1}{\alab_2}{\theload{\areg}{\anexpr}}}&:=&
	\thetransition{\alab_1}{\tilde\alab_{x}}{\thecondition{\themem{\auxaccesslevel{\anexpr}}=\storeacc}}\label{Equation:HelperLoadMove}\\
	&&\thetransition{\tilde\alab_x}{\tilde\alab_2}{\theload{\areg}{\anexpr}}\notag\\[1mm]
	%
	\semhelpertrans{\thetransition{\alab_1}{\alab_2}{\thestore{\anexpr_1}{\anexpr_2}}}&:=&
	\thetransition{\alab_1}{\tilde\alab_{x1}}{\thecondition{\themem{\auxaccesslevel{\anexpr_1}}\geq \loadacc}}\label{Equation:HelperStoreMove}\\
	&&\thetransition{\tilde\alab_{x1}}{\tilde\alab_{x2}}{\thestore{\anexpr_1}{\anexpr_2}}\notag\\
	&&\thetransition{\tilde\alab_{x2}}{\tilde\alab_{2}}{\thestore{\auxaccesslevel{\anexpr_1}}{\storeacc}}\notag\\[1mm]
	%
	\semhelpercpy{\thetransition{\alab_1}{\alab_2}{\textit{local/scfence/fence}}} &:=& \thetransition{\tilde\alab_1}{\tilde\alab_2}{\textit{local/scfence/fence}}\label{Equation:HelperLocalFence}\\[1mm]
	%
	\semhelpercpy{\thetransition{\alab_1}{\alab_2}{\thestore{\anexpr_1}{\anexpr_2}}}&:=&
	\thetransition{\tilde\alab_1}{\tilde\alab_x}{\thestore{\anexpr_1}{\anexpr_2}}\label{Equation:HelperStore}\\
	&&\thetransition{\tilde\alab_x}{\tilde\alab_2}{\thestore{\auxaccesslevel{\anexpr_1}}{\storeacc}}\notag\\[1mm]
	%
	\semhelpercpy{\thetransition{\alab_1}{\alab_2}{\theload{\areg}{\anexpr}}}&:=&
	\thetransition{\tilde\alab_1}{\tilde\alab_{x1}}{\thelocal{\tilde\areg}{\anexpr}}\label{Equation:HelperLoad}\\
	&&\thetransition{\tilde\alab_{x1}}{\tilde\alab_{x2}}{\theload{\areg}{\tilde\areg}}\notag\\
	&&\thetransition{\tilde\alab_{x2}}{\tilde\alab_2}{\thestore{\auxaccesslevel{\tilde\areg}}{\maxfun\{\loadacc, \themem{\auxaccesslevel{\tilde\areg}}\}}}\notag
	\end{eqnarray}
	\caption{Instrumentation of helpers.\label{Figure:HelpersInstrumentation}}
\end{figure}

\section{Evaluation}
We have implemented our instrumentation on top of Trencher~\cite{BDM13} and run the resulting tool on a set of example programs. Figure~\ref{Figure:porlive} shows the evaluation with partial-order reduction and liveness-based optimization. 
Figure~\ref{Figure:npornlive} shows the evaluation with the optimizations turned-off. 
The tables show how many threads, labels, and instructions the programs have, whether or not they were detected as robust, and the number of visited states for the singularity-based and locality-based analysis.
The time for the SC-reachability analysis was negligible ($<8$ seconds) in all examples.

Depending on the input program, singularity cuts down the search space to a third as compared with locality (Dekker: From 121 states with locality only 43 remain with singularity).
More importantly, this is the first (and complete) robustness analyis for PGAS programs.

\begin{figure}
\begin{tabular}{|l||r|r|r||c|r|r|}
\hline
Program&T&L&I&Robust?&Visited States (Singularity)&Visited States (Locality)\\
\hline
\hline
cilk-the-wsq-wrongdoing.txt&5&80&79&No&443&513\\
\hline
Cilk's THE WSQ (non-robust)&3&73&72&Yes&71&122\\
\hline
CLH Lock (robust)&3&42&41&No&1124&1358\\
\hline
Dekker (robust)&2&28&34&No&31&59\\
\hline
Dekker (non-robust)&2&24&30&No&35&83\\
\hline
Lamport (robust)&3&39&42&Yes&12868&12868\\
\hline
Lamport (non-robust)&3&33&36&No&142&162\\
\hline
Lock-Free Stack (robust)&4&46&50&Yes&796&820\\
\hline
MCS Lock (robust)&2&54&58&No&150&160\\
\hline
mfa.txt&3&14&11&No&122&140\\
\hline
mp.txt&2&6&4&No&20&23\\
\hline
\end{tabular}
\caption{Singularity vs. locality with partial-order reduction and liveness.}
\label{Figure:porlive}
\end{figure}

\begin{figure}
\begin{tabular}{|l||r|r|r||c|r|r|}
\hline
Program&T&L&I&Robust?&Visited States (Singularity)&Visited States (Locality)\\
\hline
\hline
cilk-the-wsq-wrongdoing.txt&5&80&79&No&532&575\\
\hline
Cilk's THE WSQ (non-robust)&3&73&72&Yes&95&150\\
\hline
CLH Lock (robust)&3&42&41&No&1289&1496\\
\hline
Dekker (robust)&2&28&34&No&37&103\\
\hline
Dekker (non-robust)&2&24&30&No&43&121\\
\hline
Lamport (robust)&3&39&42&Yes&77657&77657\\
\hline
Lamport (non-robust)&3&33&36&No&159&185\\
\hline
Lock-Free Stack (robust)&4&46&50&Yes&5504&5504\\
\hline
MCS Lock (robust)&2&54&58&No&169&176\\
\hline
mfa.txt&3&14&11&No&124&146\\
\hline
mp.txt&2&6&4&No&22&25\\
\hline
\end{tabular}
\caption{Singularity vs. locality without partial-order reduction and liveness.}
\label{Figure:npornlive}
\end{figure}

\end{document}